\documentclass[a4paper,reqno]{amsart}

\usepackage{amssymb}
\usepackage{amsfonts}
\usepackage{latexsym}
\usepackage{amsmath}
\usepackage{amsthm}
\usepackage[mathcal]{euscript}
\usepackage{bbm}
\usepackage{tikz}

\def\sA{{\mathfrak A}}   \def\sB{{\mathfrak B}}   \def\sC{{\mathfrak C}}

\def\cA{{\mathcal A}}

   \def\cK{{\mathcal K}}

\def\cal H{{\mathcal H}}

\def\R{\mathbb{R}}
\def\C{\mathbb{C}}

\def\N{\mathbb{N}}

\def\ran{{\text{\rm ran\,}}}
\def\dom{{\text{\rm dom\,}}}

\def\phi{\varphi}

\def\eps{\varepsilon}

\def\fa{\mathfrak{a}}

\DeclareMathOperator{\Real}{Re}

\def\eins{\mathbbm{1}}

\newtheorem{theorem}{Theorem}[section]
\newtheorem{proposition}[theorem]{Proposition}
\newtheorem{corollary}[theorem]{Corollary}
\newtheorem{lemma}[theorem]{Lemma}

\theoremstyle{definition}
\newtheorem{definition}[theorem]{Definition}
\newtheorem{example}[theorem]{Example}
\newtheorem{assumption}[theorem]{Assumption}

\newtheorem{remark}[theorem]{Remark}
\newtheorem*{acknowledgments}{Acknowledgements}

\numberwithin{equation}{section}

\begin{document}

\title[]{Generalized interactions supported on hypersurfaces}

\author{Pavel Exner}
\address{Department of Theoretical Physics, Nuclear Physics Institute, Czech Academy of Sciences, 25068 \v Re\v z near Prague, Czechia, and Doppler Institute for Mathematical Physics and Applied Mathematics, Czech Technical University, B\v rehov\'a 7, 11519 Prague, Czechia}
\email{exner@ujf.cas.cz} 

\author{Jonathan Rohleder}
\address{Institut f\"ur Numerische Mathematik, Technische Universit\"at Graz, Steyrergasse 30, A-8010 Graz, Austria}
\email{rohleder@tugraz.at}


\begin{abstract}
We analyze a family of singular Schr\"odinger operators with local singular interactions supported by a hypersurface $\Sigma \subset \R^n, \: n\ge 2$, being the boundary of a Lipschitz domain, bounded or unbounded, not necessarily connected. At each point of $\Sigma$ the interaction is characterized by four real parameters, the earlier studied case of $\delta$- and $\delta'$-interactions being particular cases. We discuss spectral properties of these operators and derive operator inequalities between those referring to the same hypersurface but different couplings and describe their implications for spectral properties.
\end{abstract}

\maketitle

\section{Introduction}

Solvable models play an important role in our understanding of quantum systems because they often allow us to describe their properties through tools which are simplified but better accessible mathematically. This applies, in particular, to the so-called leaky quantum graphs and similar structures, cf.~\cite[Chap.~10]{EK15}, describing the motion of quantum particles confined to `thin' regions of space of a possibly nontrivial geometrical and topological character in a way which does not neglect the tunneling effect.

Technically speaking, such models emerged as a natural generalization of point interaction systems \cite{AGHH05}. First examples appeared about a quarter of century ago \cite{AGS87, BT92}, but more attention to these problems was attracted only later, and at present we experience a renewed interest to them. At the beginning measure-type potentials, usually dubbed $\delta$-type, were studied, and recently also more singular ones came into focus \cite{BEL14, JL15}. The $\delta'$-interaction represents an interesting object especially in view of its scattering properties: in contrast to more regular potentials a $\delta'$-barrier becomes more opaque as the energy increases which opens ways to unexpected physical effects \cite{AEL94}. What is important is that the $\delta'$-interaction is more than a mere mathematical construct, because it can be approximated in any fixed interval of energies by families of regular potentials following the seminal idea of Cheon and Shigehara \cite{CS98}.

The $\delta$- and $\delta'$-interactions are not the whole story, though. It follows from basic facts about self-adjoint extensions that the most general point interaction is characterized by four parameters, and \emph{mutatis mutandis}, four functions are needed to describe such a general singular interaction supported by a manifold of codimension one. A discussion of such interactions is the main topic of this paper. Speaking of the motivation, we note that it is again more than a mere mathematical extension, because these interactions include another type of scattering behavior, different from both the $\delta$- and $\delta'$-situations, as one can see, for instance, from the high-energy resonance asymptotics of the corresponding generalization of the so-called Winter model \cite{EF06}.

There is not much in the literature about such general singular interactions supported by hypersurfaces beyond examples with a symmetry allowing for a dimension reduction. Our present study overlaps with fresh results in~\cite{MPS15} but there are considerable differences. The said paper considers singular perturbations of more general elliptic operators with nontrivial coefficients in the second and zeroth order but, on the other hand, it is restricted to compact hypersurfaces, while in the present paper we focus on the Laplacian but allow a more general class of possibly noncompact hypersurfaces. The main results in~\cite{MPS15} concern the description of all selfadjoint boundary conditions as well as Krein type resolvent formulas and Schatten-von Neumann properties, in contrast to the present paper, where a class of local interactions is fixed and spectral properties of this class are in the center of interest. A further substantial difference is the used technique. The paper of Mantile et al.\ approaches the problem through resolvents of the involved operators, while we define the singular perturbations using the appropriate quadratic form.

Let us describe briefly the contents of this paper. In the next section we introduce the singular Schr\"odinger operators and show how they can be described alternatively by boundary conditions on the hypersurface, writing the latter in the form introduced first in~\cite{EG99} in the one-dimensional case. Its advantage is that one can simply distinguish the particular cases of $\delta$- and $\delta'$-interactions. In Section~3 we show that in the case of a compact interaction the essential spectrum is the positive halfline and the discrete one is finite, and we also find conditions under which the latter is empty or nonempty. Furthermore, for interactions supported on a noncompact, asymptotically planar hypersurface in $\R^3$ we obtain a lower bound for the essential spectrum. In Section~4 we derive operator inequalities between different elements of the considered operator family, generalizing those between the $\delta$- and $\delta'$-interactions found in~\cite{BEL14} and use them to establish further spectral results.

\section{Generalized interactions and quadratic forms}

In this section we introduce the generalized interactions under consideration via quadratic forms and investigate their action and domain.

Let $\Sigma \subset \R^n$, $n \geq 2$, be the boundary of a (bounded or unbounded, not necessarily connected) Lipschitz domain $\Omega = \Omega_{\rm i}$ and let $\Omega_{\rm e} = \R^n \setminus (\Omega_{\rm i} \cup \Sigma)$. In the following we denote by $H^s (\Omega_j)$, $j = \rm i, e$, and $H^s (\R^n)$ the Sobolev spaces of order $s \geq 0$, by $H^{s} (\Sigma)$ the Sobolev space of order $s \in [0, 1]$ on $\Sigma$ and by $H^{- s} (\Sigma)$ its dual. For $f \in L^2 (\R^n)$ we use to write $f_j = f |_{\Omega_j}$, $j =\rm i, e$, and $f = f_{\rm i} \oplus f_{\rm e}$. We denote the trace of a function $f \in H^1 (\Omega_j)$ on $\Sigma$ by $f |_\Sigma \in H^{1/2} (\Sigma)$. Moreover, for each $f \in H^1 (\Omega_j)$ such that $\Delta f$, calculated as a distribution, belongs to $L^2 (\Omega_j)$ we define the derivative of $f$ with respect to the outer unit normal on $\Sigma = \partial \Omega_j$ as the unique element $\partial_{\nu_j} f |_\Sigma \in H^{- 1/2} (\Sigma)$ which satisfies Green's first identity 
\begin{align}\label{eq:Green}
 \int_{\Omega_j} \nabla f \cdot \overline{\nabla g} \, d x + \int_{\Omega_j} \Delta f \overline{g} \, d x = (\partial_{\nu_j} f |_\Sigma, g |_\Sigma)_\Sigma,
\end{align}
where $(\cdot, \cdot)_\Sigma$ denotes the sesquilinear duality of $H^{- 1/2} (\Sigma)$ and $H^{1/2} (\Sigma)$. Note that if $\Sigma$ is sufficiently smooth and $f$ is differentiable up to the boundary then $\partial_{\nu_j} f |_\Sigma$ is the usual derivative with respect to the outer unit normal. We also remark that the outer unit normals for $\Omega_{\rm i}$ and $\Omega_{\rm e}$ coincide up to a minus sign. In particular, for $f \in H^2 (\R^n)$ we have $\partial_{\nu_{\rm i}} f_{\rm i} |_\Sigma + \partial_{\nu_{\rm e}} f_{\rm e} |_\Sigma = 0$.

Throughout this paper we make the following assumption.

\begin{assumption}\label{as:coefficients}
Assume that $\alpha : \Sigma \to \R$ and $\gamma : \Sigma \to \C$ are bounded, measurable functions. Moreover, let $\Sigma_\beta \subset \Sigma$ be a relatively open subset and let $\beta : \Sigma \to \R$ be a function such that $1/\beta$ is measurable and bounded on $\Sigma_\beta$ and $\beta = 0$ identically on $\Sigma_0 := \Sigma \setminus \Sigma_\beta$. 
\end{assumption}

We focus on generalized interactions supported on $\Sigma$ described by the (negative) Laplacian on $\R^n \setminus \Sigma$ subject to the interface conditions
\begin{align}\label{eq:conditions}
\begin{split}
  \partial_{\nu_{\rm i}} f_{\rm i} |_\Sigma + \partial_{\nu_{\rm e}} f_{\rm e} |_\Sigma & = \frac{\alpha}{2} \left( f_{\rm i} |_\Sigma + f_{\rm e} |_\Sigma \right) + \frac{\gamma}{2} \left( \partial_{\nu_{\rm i}} f_{\rm i} |_\Sigma - \partial_{\nu_{\rm e}} f_{\rm e} |_\Sigma \right), \\
 f_{\rm i} |_\Sigma - f_{\rm e} |_\Sigma & = - \frac{\overline{\gamma}}{2} \left( f_{\rm i} |_\Sigma + f_{\rm e} |_\Sigma \right) + \frac{\beta}{2} \left( \partial_{\nu_{\rm i}} f_{\rm i} |_\Sigma - \partial_{\nu_{\rm e}} f_{\rm e} |_\Sigma \right)
\end{split}
\end{align}
on $\Sigma$. Observe that for the time being the conditions~\eqref{eq:conditions} are formal; for instance, $\partial_{\nu_{\rm i}} f_{\rm i} |_\Sigma - \partial_{\nu_{\rm e}} f_{\rm e} |_\Sigma$ belongs to $H^{- 1/2} (\Sigma)$ and its multiplication by $\beta$ or $\gamma$ does not make sense if these coefficients are very irregular. In fact, we are going to define a Hamiltonian first by use of an appropriate quadratic form and we will show afterwards that under a minor, reasonable additional assumption the functions in its domain satisfy~\eqref{eq:conditions}; cf.\ Theorem~\ref{thm:HA} below. We remark that the conditions~\eqref{eq:conditions} include generalized interactions on nonclosed hypersurfaces being a subset of a closed hypersurface, by choosing $\alpha, \beta$ and $\gamma$ as zero on a part of $\Sigma$.

In the following we write 
\begin{align}\label{eq:A}
 \cA = \begin{pmatrix} \alpha & \gamma \\ - \overline{\gamma} & \beta \end{pmatrix}.
\end{align}
We define the symmetric matrix function $\Theta_\cA$ on $\Sigma$ by
\begin{align}\label{eq:Theta}
 \Theta_\cA = \begin{pmatrix}
           \frac{|1 + \frac{\gamma}{2}|^2}{\beta} \eins_{\Sigma_\beta} + \frac{\alpha}{4} & \frac{(\frac{\overline{\gamma}}{2} - 1) (1 + \frac{\gamma}{2})}{\beta} \eins_{\Sigma_\beta} + \frac{\alpha}{4}\\
	   \frac{(\frac{\gamma}{2} - 1) (1 + \frac{\overline{\gamma}}{2})}{\beta} \eins_{\Sigma_\beta} + \frac{\alpha}{4}  & \frac{|1 - \frac{\gamma}{2}|^2}{\beta} \eins_{\Sigma_\beta} + \frac{\alpha}{4}
          \end{pmatrix},
\end{align}
where $\frac{1}{\beta} \eins_{\Sigma_\beta}$ equals $1 / \beta$ on $\Sigma_\beta$ and zero on $\Sigma_0$. Moreover, we define a quadratic form $h_\cA$ in $L^2 (\R^n)$ via
\begin{align}\label{eq:form}
\begin{split}
 h_\cA [f, g] & = \int_{\Omega_{\rm i}} \nabla f_{\rm i} \cdot \overline{\nabla g_{\rm i}} \, d x + \int_{\Omega_{\rm e}} \nabla f_{\rm e} \cdot \overline{\nabla g_{\rm e}} \, d x - \int_\Sigma \left\langle \Theta_\cA \binom{f_{\rm i}}{f_{\rm e}},  \binom{g_{\rm i}}{g_{\rm e}} \right\rangle d \sigma, \\
 \dom h_\cA & = \Big\{ f_{\rm i} \oplus f_{\rm e} \in H^1 (\Omega_{\rm i}) \oplus H^1 (\Omega_{\rm e}) : (1 + \tfrac{\overline \gamma}{2} ) f_{\rm i} = (1 - \tfrac{\overline \gamma}{2}) f_{\rm e}~\text{on}~\Sigma_0 \Big\},
\end{split}
\end{align}
where the brackets $\langle \cdot, \cdot \rangle$ denote the inner product in~$\C^2$, $\sigma$ is the standard surface measure on $\Sigma$ and the functions in the boundary integral have to be understood as the appropriate traces. Note that $h_\cA$ is well-defined since the entries of $\Theta_\cA$ are bounded functions. 

In the following lemma we investigate properties of $h_\cA$. For all details concerning semibounded quadratic forms and corresponding selfadjoint operators we refer the reader to the standard literature, e.g.,~\cite[Chap.~VI]{K95}.

\begin{lemma}\label{lem:form}
The quadratic form $h_\cA$ in $L^2 (\R^n)$ is densely defined, symmetric, semibounded below and closed.
\end{lemma}

\begin{proof}
Clearly $h_\cA$ is densely defined as its domain contains $C_0^\infty (\Omega_{\rm i}) \oplus C_0^\infty (\Omega_{\rm e})$. Note further that for each $s \in \Sigma$ the matrix $\Theta_\cA (s)$ is symmetric, which implies the symmetry of~$h_\cA$. Moreover, since $\alpha, \gamma$ and $1/\beta |_{\Sigma_\beta} \cdot \eins_{\Sigma_\beta}$ are bounded functions, there exists a constant $\eta \in \R$ (independent of $s$), without loss of generality we may suppose $\eta < 0$, such that
\begin{align}\label{eq:lowerBoundPointwise}
 - \Big\langle \Theta_\cA (s) \binom{f_{\rm i} (s)}{f_{\rm e} (s)}, \binom{f_{\rm i} (s)}{f_{\rm e} (s)} \Big\rangle \geq \eta \left( |f_{\rm i} (s)|^2 + |f_{\rm e} (s)|^2 \right), \quad s \in \Sigma.
\end{align}
Recall that by Ehrling's lemma for each $\eps > 0$ there exists $C_\eps (\Omega_j) > 0$ such that
\begin{align*}
 \| f_j |_\Sigma \|_{L^2 (\Sigma)}^2 \leq \eps \| f_j \|_{H^1 (\Omega_j)}^2 + C_\eps (\Omega_j) \| f_j \|_{L^2 (\Omega_j)}^2, \quad f \in H^1 (\Omega_j), \quad j = \rm i, e;
\end{align*}
for a proof of this inequality in the case of a Lipschitz domain with a possibly noncompact boundary see, e.g.,~\cite[Lemma~2.6]{BEL14}. Therefore it follows from~\eqref{eq:lowerBoundPointwise} that for each $\eps > 0$ there exists $C_\eps = C_\eps (\R^n) > 0$ with
\begin{align}\label{eq:lowerBound}
 h_\cA [f] \geq (1 + \eta \eps) \big( \|f_{\rm i}\|_{H^1 (\Omega_{\rm i})}^2 + \|f_{\rm e}\|_{H^1 (\Omega_{\rm e})}^2 \big) + (\eta C_\eps - 1) \|f\|_{L^2 (\R^n)}^2
\end{align}
for all $f \in \dom h_\cA$; here and in the following we use the abbreviation $h_\cA [f] = h_\cA [f, f]$. In particular, for each sufficiently small $\eps > 0$ we have $1 + \eta \eps > 0$ and~\eqref{eq:lowerBound} implies
\begin{align*}
 h_\cA [f] \geq (\eta C_\eps - 1) \|f\|_{L^2 (\R^n)}^2, \quad f \in \dom h_\cA.
\end{align*}
Thus $h_\cA$ is semibounded below. Moreover, we can conclude from~\eqref{eq:lowerBound} that 
\begin{align*}
 (h_\cA - \eta C_\eps + 1) [f] \geq (1 + \eta \eps) \big( \|f_{\rm i}\|_{H^1 (\Omega_{\rm i})}^2 + \|f_{\rm e}\|_{H^1 (\Omega_{\rm e})}^2 \big), \quad f \in \dom h_\cA.
\end{align*}
From this, the boundedness assumptions on $\alpha, \beta$ and $\gamma$ and the continuity of the trace map from $H^1 (\Omega_j)$ to $L^2 (\Sigma)$, $j = \rm i, e$, it follows that the norm induced by $h_\cA - \eta C_\eps + 1$ is equivalent to the norm in $H^1 (\Omega_{\rm i}) \oplus H^1 (\Omega_{\rm e})$. Since $\dom h_\cA$ is a closed subspace of $H^1 (\Omega_{\rm i}) \oplus H^1 (\Omega_{\rm e})$, it follows that $h_\cA$ is closed.
\end{proof}

The previous lemma allows us to define a selfadjoint Hamiltonian in the following way.

\begin{definition}\label{def:HA}
The Laplacian subject to the generalized interaction~\eqref{eq:conditions} on $\Sigma$ is defined as the selfadjoint, semibounded operator $- \Delta_\cA$ in $L^2 (\R^n)$ corresponding to the quadratic form $h_\cA$ in~\eqref{eq:form}, i.e., $- \Delta_{\cA}$ is the unique selfadjoint operator in $L^2 (\R^n)$ which satisfies $\dom (- \Delta_\cA) \subset \dom h_\cA$ and
\begin{align*}
 (- \Delta_\cA f, g)_{L^2 (\R^n)} = h_\cA [f, g]
\end{align*}
for all $f \in \dom (- \Delta_\cA)$, $g \in \dom h_\cA$.
\end{definition}

In the following theorem we establish the relation of $- \Delta_\cA$ with the conditions~\eqref{eq:conditions}. As mentioned above, the conditions~\eqref{eq:conditions} are only formal and need an interpretation particularly if $\beta$ and $\gamma$ are nonsmooth. In order to give a meaning to the first condition in~\eqref{eq:conditions}, in the following theorem {\it we assume additionally} that $\gamma$ satisfies
\begin{align}\label{eq:gamma}
 \gamma \phi \in H^{1/2} (\Sigma) \quad \text{for~all}~\phi \in H^{1/2} (\Sigma),
\end{align}
which allows us to define $\gamma \psi$ for each $\psi \in H^{- 1/2} (\Sigma)$ by $(\gamma \psi, \phi)_\Sigma = (\psi, \overline{\gamma} \phi)_\Sigma$ for all $\phi \in H^{1/2} (\Sigma)$. A rigorous formulation of~\eqref{eq:conditions} is then given by the conditions
\begin{align}\label{eq:conditionsNew}
\begin{split}
 \partial_{\nu_{\rm i}} f_{\rm i} |_\Sigma + \partial_{\nu_{\rm e}} f_{\rm e} |_\Sigma & = \frac{\alpha}{2} \left( f_{\rm i} |_\Sigma + f_{\rm e} |_\Sigma \right) + \frac{\gamma}{2} \left( \partial_{\nu_{\rm i}} f_{\rm i} |_\Sigma - \partial_{\nu_{\rm e}} f_{\rm e} |_\Sigma \right), \\
 f_{\rm i} |_{\Sigma_0} - f_{\rm e} |_{\Sigma_0} & = - \frac{\overline{\gamma}}{2} \left( f_{\rm i} |_{\Sigma_0} + f_{\rm e} |_{\Sigma_0} \right), \\
 \frac{1}{\beta} \left(f_{\rm i} |_{\Sigma_\beta} - f_{\rm e} |_{\Sigma_\beta} \right) & = - \frac{\overline{\gamma}}{2 \beta} \left( f_{\rm i} |_{\Sigma_\beta} + f_{\rm e} |_{\Sigma_\beta} \right) + \frac{1}{2} \left( \partial_{\nu_{\rm i}} f_{\rm i} |_{\Sigma_\beta} - \partial_{\nu_{\rm e}} f_{\rm e} |_{\Sigma_\beta} \right),
\end{split}
\end{align}
where the latter equality is to be understood in the sense of distributions, namely,
\begin{align*}
 \Big(\frac{1}{\beta} \eins_{\Sigma_\beta} \left(f_{\rm i} |_\Sigma - f_{\rm e} |_\Sigma \right), \phi \Big)_\Sigma & = \Big( - \frac{\overline{\gamma}}{2 \beta} \eins_{\Sigma_\beta} \left( f_{\rm i} |_\Sigma + f_{\rm e} |_\Sigma \right) + \frac{1}{2} \left( \partial_{\nu_{\rm i}} f_{\rm i} |_\Sigma - \partial_{\nu_{\rm e}} f_{\rm e} |_\Sigma \right), \phi \Big)_\Sigma
\end{align*}
for all $\phi \in H^{1/2} (\Sigma)$ such that $\phi |_{\Sigma_0} = 0$. Note that if, e.g., $\beta$ is constant on $\Sigma$ then~\eqref{eq:conditions} makes sense and is equivalent to~\eqref{eq:conditionsNew}.

\begin{theorem}\label{thm:HA}
Let Assumption~\ref{as:coefficients} be satisfied and assume in addition that $\gamma$ satisfies~\eqref{eq:gamma}. Then the selfadjoint operator $- \Delta_\cA$ in Definition~\ref{def:HA} is given by
\begin{align}\label{eq:op}
\begin{split}
 - \Delta_\cA f & = - \Delta f_{\rm i} \oplus - \Delta f_{\rm e}, \\
 \dom \big( - \Delta_\cA \big) & = \big\{ f = f_{\rm i} \oplus f_{\rm e} \in H^1 (\Omega_{\rm i}) \oplus H^1 (\Omega_{\rm e}) : \Delta f_j \in L^2 (\Omega_j), j = {\rm i, e}, \\ & \qquad f~\text{satisfies}~\eqref{eq:conditionsNew} \big\}.
\end{split}
\end{align}
\end{theorem}

\begin{proof}
Let us denote by $H$ the operator given in~\eqref{eq:op}. In order to show that $- \Delta_\cA = H$ let first $f \in \dom H$. Then $f \in H^1 (\Omega_{\rm i}) \oplus H^1 (\Omega_{\rm e})$ and $\Delta f_j \in L^2 (\Omega_j)$, $j = \rm i, e$, and the second identity in~\eqref{eq:conditionsNew} immediately implies $f \in \dom h_\cA$. Furthermore, for $g \in \dom h_\cA$ Green's identity~\eqref{eq:Green} yields 
\begin{align}\label{eq:longCalc}
\begin{split}
 (H f, g)_{L^2 (\R^n)} & = \int_{\Omega_{\rm i}} \nabla f_{\rm i} \cdot \overline{\nabla g_{\rm i}} \, d x + \int_{\Omega_{\rm e}} \nabla f_{\rm e} \cdot \overline{\nabla g_{\rm e}} \, d x \\
 & \quad - \left( \partial_{\nu_{\rm i}} f_{\rm i} |_\Sigma, g_{\rm i} |_\Sigma \right)_\Sigma - \left(\partial_{\nu_{\rm e}} f_{\rm e} |_\Sigma, g_{\rm e} |_\Sigma \right)_\Sigma.
\end{split}
\end{align}
On the other hand, for $g \in \dom h_\cA$ the fact that $(1 + \frac{\overline{\gamma}}{2}) g_{\rm i} |_\Sigma - (1 - \frac{\overline{\gamma}}{2}) g_{\rm e} |_\Sigma$ vanishes on $\Sigma_0$ and the third identity in~\eqref{eq:conditionsNew} yield
\begin{align}\label{eq:boundaryTerm}
\begin{split}
 \int_\Sigma \big\langle \Theta_\cA \tbinom{f_{\rm i}}{f_{\rm e}}, \tbinom{g_{\rm i}}{g_{\rm e}} \big\rangle d \sigma & = \int_{\Sigma} \tfrac{1}{\beta} \eins_{\Sigma_\beta} \big( (1 + \tfrac{\overline{\gamma}}{2}) f_{\rm i} - (1 - \tfrac{\overline{\gamma}}{2}) f_{\rm e} \big) \big( \overline{(1 + \tfrac{\overline{\gamma}}{2}) g_{\rm i} - (1 - \tfrac{\overline{\gamma}}{2}) g_{\rm e}} \big) d \sigma \\
 & \quad + \int_\Sigma \tfrac{\alpha}{4} (f_{\rm i} + f_{\rm e}) (\overline{g_{\rm i} + g_{\rm e}}) d \sigma \\
 & = \tfrac{1}{2} \big( \partial_{\nu_{\rm i}} f_{\rm i} |_{\Sigma} - \partial_{\nu_{\rm e}} f_{\rm e} |_{\Sigma}, (1 + \tfrac{\overline{\gamma}}{2}) g_{\rm i} |_\Sigma - (1 - \tfrac{\overline{\gamma}}{2}) g_{\rm e} |_\Sigma \big)_\Sigma \\
 & \quad + \tfrac{1}{2} \left( \tfrac{\alpha}{2} (f_{\rm i} |_\Sigma + f_{\rm e} |_\Sigma), g_{\rm i} |_\Sigma + g_{\rm e} |_\Sigma \right)_\Sigma.
\end{split}
\end{align}
Furthermore, with an application of the first identity in~\eqref{eq:conditionsNew},~\eqref{eq:boundaryTerm} turns into
\begin{align*}
\begin{split}
 \int_\Sigma \big\langle \Theta_\cA \tbinom{f_{\rm i}}{f_{\rm e}}, \tbinom{g_{\rm i}}{g_{\rm e}} \big\rangle d \sigma & = \tfrac{1}{2} \big( \partial_{\nu_{\rm i}} f_{\rm i} |_{\Sigma} - \partial_{\nu_{\rm e}} f_{\rm e} |_{\Sigma}, (1 + \tfrac{\overline{\gamma}}{2}) g_{\rm i} |_\Sigma - (1 - \tfrac{\overline{\gamma}}{2}) g_{\rm e} |_\Sigma \big)_\Sigma \\
 & \quad + \tfrac{1}{2} \left( (1 - \tfrac{\gamma}{2}) \partial_{\nu_{\rm i}} f_{\rm i} |_{\Sigma} + (1 + \tfrac{\gamma}{2}) \partial_{\nu_{\rm e}} f_{\rm e} |_{\Sigma},  g_{\rm i} |_\Sigma + g_{\rm e} |_\Sigma \right)_\Sigma \\
 & = \tfrac{1}{2} \big( (1 + \tfrac{\gamma}{2}) (\partial_{\nu_{\rm i}} f_{\rm i} |_{\Sigma} - \partial_{\nu_{\rm e}} f_{\rm e} |_{\Sigma}) \\
 & \qquad + (1 - \tfrac{\gamma}{2}) \partial_{\nu_{\rm i}} f_{\rm i} |_{\Sigma} + (1 + \tfrac{\gamma}{2}) \partial_{\nu_{\rm e}} f_{\rm e} |_{\Sigma}, g_{\rm i} |_\Sigma \big)_\Sigma \\
 & \quad + \tfrac{1}{2} \big( (\tfrac{\gamma}{2} - 1) (\partial_{\nu_{\rm i}} f_{\rm i} |_{\Sigma} - \partial_{\nu_{\rm e}} f_{\rm e} |_{\Sigma}) \\
 & \qquad + (1 - \tfrac{\gamma}{2}) \partial_{\nu_{\rm i}} f_{\rm i} |_{\Sigma} + (1 + \tfrac{\gamma}{2}) \partial_{\nu_{\rm e}} f_{\rm e} |_{\Sigma}, g_{\rm e} |_\Sigma \big)_\Sigma \\
 & = (\partial_{\nu_{\rm i}} f_{\rm i} |_{\Sigma}, g_{\rm i} |_\Sigma)_\Sigma + (\partial_{\nu_{\rm e}} f_{\rm e} |_{\Sigma}, g_{\rm e} |_\Sigma)_\Sigma.
\end{split}
\end{align*}
From this and~\eqref{eq:longCalc} it follows 
\begin{align*}
 (H f, g)_{L^2 (\R^n)} = h_\cA [f, g]
\end{align*}
for all $g \in \dom h_\cA$. Thus $f \in \dom (- \Delta_\cA)$ and $- \Delta_\cA f = H f$.

It remains to show that each $f \in \dom (- \Delta_\cA)$ belongs to $\dom H$. Indeed, choosing such $f$ we have $f \in H^1 (\Omega_{\rm i}) \oplus H^1 (\Omega_{\rm e})$ and for $g \in C_0^\infty (\Omega_{\rm i}) \oplus C_0^\infty (\Omega_{\rm e}) \subset \dom h_\cA$
\begin{align*}
 (- \Delta_\cA f, g)_{L^2 (\R^n)} & = h_\cA [f, g] = (- \Delta f_{\rm i}, g_{\rm i}) + (- \Delta f_{\rm e}, g_{\rm e}),
\end{align*}
where the expressions on the right mean a distributional application. This implies 
\begin{align}\label{eq:LaplaceAction}
 - \Delta f_{\rm i} \oplus - \Delta f_{\rm e} = - \Delta_\cA f \in L^2 (\R^n) = L^2 (\Omega_{\rm i}) \oplus L^2 (\Omega_{\rm e}).
\end{align}
In order to verify the boundary conditions~\eqref{eq:conditionsNew} note that by the choice of $f$ we have
\begin{align}\label{eq:BCpart1}
 (1 + \overline{\gamma}/2) f_{\rm i} |_{\Sigma_0} = (1 - \overline{\gamma}/2) f_{\rm e} |_{\Sigma_0},
\end{align}
which is the second condition in~\eqref{eq:conditionsNew}. Moreover,~\eqref{eq:LaplaceAction} and Green's identity yield
\begin{align}\label{eq:f1}
\begin{split}
 (- \Delta_\cA f, g)_{L^2 (\R^n)} & = \int_{\Omega_{\rm i}} \nabla f_{\rm i} \overline{\nabla g_{\rm i}} \, d x + \int_{\Omega_{\rm e}} \nabla f_{\rm e} \overline{\nabla g_{\rm e}} \, d x \\
 & \quad - \left( \partial_{\nu_{\rm i}} f_{\rm i} |_\Sigma, g_{\rm i} |_\Sigma \right)_\Sigma - \left(\partial_{\nu_{\rm e}} f_{\rm e} |_\Sigma, g_{\rm e} |_\Sigma \right)_\Sigma
\end{split}
\end{align}
for all $g \in \dom h_\cA$ and, on the other hand, using the definition of~$- \Delta_\cA$,
\begin{align}\label{eq:f2}
\begin{split}
 (- \Delta_\cA f, g)_{L^2 (\R^n)} & = \int_{\Omega_{\rm i}} \nabla f_{\rm i} \overline{\nabla g_{\rm i}} \, d x + \int_{\Omega_{\rm e}} \nabla f_{\rm e} \overline{\nabla g_{\rm e}} \, d x - \int_\Sigma \big\langle \Theta_\cA \tbinom{f_{\rm i}}{f_{\rm e}}, \tbinom{g_{\rm i}}{g_{\rm e}} \big\rangle d \sigma
\end{split}
\end{align}
for all $g \in \dom h_\cA$. From~\eqref{eq:f1} and~\eqref{eq:f2} we conclude
\begin{align}\label{eq:f3}
 \left( \partial_{\nu_{\rm i}} f_{\rm i} |_\Sigma, g_{\rm i} |_\Sigma \right)_\Sigma + \left(\partial_{\nu_{\rm e}} f_{\rm e} |_\Sigma, g_{\rm e} |_\Sigma \right)_\Sigma & = \int_\Sigma \left\langle \Theta_\cA \binom{f_{\rm i} |_\Sigma}{f_{\rm e} |_\Sigma}, \binom{g_{\rm i} |_\Sigma}{g_{\rm e} |_\Sigma} \right\rangle d \sigma
\end{align}
for all $g \in \dom h_\cA$. From this identity we are going to derive the conditions~\eqref{eq:conditionsNew}. In fact, for each $g \in \dom h_\cA$~\eqref{eq:f3} can be rewritten as
\begin{align}\label{eq:veryLongCalc}
\begin{split}
 0 & = \left( \partial_{\nu_{\rm i}} f_{\rm i} |_\Sigma - (1 + \tfrac{\gamma}{2}) \big( \tfrac{1 + \overline \gamma / 2}{\beta} \eins_{\Sigma_\beta} f_{\rm i} |_\Sigma + \tfrac{\overline \gamma / 2 - 1}{\beta} \eins_{\Sigma_\beta} f_{\rm e} |_\Sigma \big) , g_{\rm i} |_\Sigma \right)_\Sigma \\
 & \quad + \left( \partial_{\nu_{\rm e}} f_{\rm e} |_\Sigma - (\tfrac{\gamma}{2} - 1) \big( \tfrac{1 + \overline \gamma / 2}{\beta} \eins_{\Sigma_\beta} f_{\rm i} |_\Sigma + \tfrac{\overline \gamma / 2 - 1}{\beta} \eins_{\Sigma_\beta} f_{\rm e} |_\Sigma \big), g_{\rm e} |_\Sigma \right)_\Sigma \\
 & \quad - \left( \tfrac{\alpha}{4} (f_{\rm i} |_\Sigma + f_{\rm e} |_\Sigma), g_{\rm i} |_\Sigma + g_{\rm e} |_\Sigma \right)_\Sigma \\
 & = \left(  \tfrac{1 + \overline \gamma / 2}{\beta} \eins_{\Sigma_\beta} f_{\rm i} |_\Sigma + \tfrac{\overline \gamma / 2 - 1}{\beta} \eins_{\Sigma_\beta} f_{\rm e} |_\Sigma, (1 - \tfrac{\overline{\gamma}}{2}) g_{\rm e} |_\Sigma - (1 + \tfrac{\overline{\gamma}}{2}) g_{\rm i} |_\Sigma \right)_\Sigma \\
 & \quad + \left( \partial_{\nu_{\rm i}} f_{\rm i} |_\Sigma, g_{\rm i} |_\Sigma \right)_\Sigma + \left( \partial_{\nu_{\rm e}} f_{\rm e} |_\Sigma, g_{\rm e} |_\Sigma \right)_\Sigma - \left( \tfrac{\alpha}{4} (f_{\rm i} |_\Sigma + f_{\rm e} |_\Sigma), g_{\rm i} |_\Sigma + g_{\rm e} |_\Sigma \right)_\Sigma \\
 & = \left(  \tfrac{1 + \overline \gamma / 2}{\beta} \eins_{\Sigma_\beta} f_{\rm i} |_\Sigma + \tfrac{\overline \gamma / 2 - 1}{\beta} \eins_{\Sigma_\beta} f_{\rm e} |_\Sigma, (1 - \tfrac{\overline{\gamma}}{2}) g_{\rm e} |_\Sigma - (1 + \tfrac{\overline{\gamma}}{2}) g_{\rm i} |_\Sigma \right)_\Sigma \\
 & \quad + \tfrac{1}{2} \left( (1 + \tfrac{\gamma}{2}) \partial_{\nu_{\rm i}} f_{\rm i} |_\Sigma + (1 - \tfrac{\gamma}{2}) \partial_{\nu_{\rm i}} f_{\rm i} |_\Sigma, g_{\rm i} |_\Sigma \right)_\Sigma \\
 & \quad + \tfrac{1}{2} \left( (1 - \tfrac{\gamma}{2}) \partial_{\nu_{\rm e}} f_{\rm e} |_\Sigma + (1 + \tfrac{\gamma}{2}) \partial_{\nu_{\rm e}} f_{\rm e} |_\Sigma, g_{\rm e} |_\Sigma \right)_\Sigma \\
 & \quad - \left( \tfrac{\alpha}{2} (f_{\rm i} |_\Sigma + f_{\rm e} |_\Sigma), \tfrac{1}{2} (g_{\rm i} |_\Sigma + g_{\rm e} |_\Sigma) \right)_\Sigma.
\end{split}
\end{align}
Let now $g \in H^1 (\Omega_{\rm i}) \oplus H^1 (\Omega_{\rm e})$ with $(1 + \frac{\overline{\gamma}}{2}) g_{\rm i} |_\Sigma = (1 - \frac{\overline{\gamma}}{2}) g_{\rm e} |_\Sigma$ on all of $\Sigma$. Then $g \in \dom h_\cA$ and~\eqref{eq:veryLongCalc} implies
\begin{align}\label{eq:nextLongCalc}
\begin{split}
 0 & = \left(  \tfrac{1 + \overline \gamma / 2}{\beta} \eins_{\Sigma_\beta} f_{\rm i} |_\Sigma + \tfrac{\overline \gamma / 2 - 1}{\beta} \eins_{\Sigma_\beta} f_{\rm e} |_\Sigma, (1 - \tfrac{\overline{\gamma}}{2}) g_{\rm e} |_\Sigma - (1 + \tfrac{\overline{\gamma}}{2}) g_{\rm i} |_\Sigma \right)_\Sigma \\ 
 & \quad + \tfrac{1}{2} \left( (1 - \tfrac{\gamma}{2}) \partial_{\nu_{\rm i}} f_{\rm i} |_\Sigma, g_{\rm e} |_\Sigma \right)_\Sigma + \tfrac{1}{2} \left( (1 - \tfrac{\gamma}{2}) \partial_{\nu_{\rm i}} f_{\rm i} |_\Sigma, g_{\rm i} |_\Sigma \right)_\Sigma \\
 & \quad + \tfrac{1}{2} \left( (1 + \tfrac{\gamma}{2}) \partial_{\nu_{\rm e}} f_{\rm e} |_\Sigma, g_{\rm i} |_\Sigma \right)_\Sigma + \tfrac{1}{2} \left( (1 + \tfrac{\gamma}{2}) \partial_{\nu_{\rm e}} f_{\rm e} |_\Sigma, g_{\rm e} |_\Sigma \right)_\Sigma \\
 & \quad - \left( \tfrac{\alpha}{2} (f_{\rm i} |_\Sigma + f_{\rm e} |_\Sigma), \tfrac{1}{2} (g_{\rm i} |_\Sigma + g_{\rm e} |_\Sigma) \right)_\Sigma \\
 & = \left(  \tfrac{1 + \overline \gamma / 2}{\beta} \eins_{\Sigma_\beta} f_{\rm i} |_\Sigma + \tfrac{\overline \gamma / 2 - 1}{\beta} \eins_{\Sigma_\beta} f_{\rm e} |_\Sigma, (1 - \tfrac{\overline{\gamma}}{2}) g_{\rm e} |_\Sigma - (1 + \tfrac{\overline{\gamma}}{2}) g_{\rm i} |_\Sigma \right)_\Sigma \\ 
 & \quad + \left( (1 - \tfrac{\gamma}{2}) \partial_{\nu_{\rm i}} f_{\rm i} |_\Sigma + (1 + \tfrac{\gamma}{2}) \partial_{\nu_{\rm e}} f_{\rm e} |_\Sigma - \tfrac{\alpha}{2} (f_{\rm i} |_\Sigma + f_{\rm e} |_\Sigma), \tfrac{1}{2} (g_{\rm i} |_\Sigma + g_{\rm e} |_\Sigma) \right)_\Sigma.
\end{split}
\end{align}
Note that each $\phi \in H^{1/2} (\Sigma)$ can be written in the form $\phi = \tfrac{1}{2} (g_{\rm i} |_\Sigma + g_{\rm e} |_\Sigma)$ for some $g \in H^1 (\Omega_{\rm i}) \oplus H^1 (\Omega_{\rm e})$ such that $(1 + \frac{\overline{\gamma}}{2}) g_{\rm i} |_\Sigma = (1 - \frac{\overline{\gamma}}{2}) g_{\rm e} |_\Sigma$. Indeed, for $\phi \in H^{1/2} (\Sigma)$ we have $(1 - \tfrac{\overline \gamma}{2}) \phi \in H^{1/2} (\Sigma)$ and $(1 + \tfrac{\overline \gamma}{2}) \phi \in H^{1/2} (\Sigma)$ by~\eqref{eq:gamma}; thus there exist $g_j \in H^1 (\Omega_j)$, $j = \rm i, e$, such that 
\begin{align*}
 g_{\rm i} |_\Sigma = (1 - \tfrac{\overline \gamma}{2}) \phi \quad \text{and} \quad g_{\rm e} |_\Sigma = (1 + \tfrac{\overline \gamma}{2}) \phi
\end{align*}
and then
\begin{align*}
 (1 + \tfrac{\overline \gamma}{2}) g_{\rm i} |_\Sigma = (1 + \tfrac{\overline \gamma}{2}) (1 - \tfrac{\overline \gamma}{2}) \phi = (1 - \tfrac{\overline \gamma}{2}) g_{\rm e} |_\Sigma
\end{align*}
holds and $\frac{1}{2} (g_{\rm i} |_\Sigma + g_{\rm e} |_\Sigma) = \phi$. Hence we obtain from~\eqref{eq:nextLongCalc}
\begin{align*}
 0 = \left( (1 - \tfrac{\gamma}{2}) \partial_{\nu_{\rm i}} f_{\rm i} |_\Sigma + (1 + \tfrac{\gamma}{2}) \partial_{\nu_{\rm e}} f_{\rm e} |_\Sigma - \tfrac{\alpha}{2} (f_{\rm i} |_\Sigma + f_{\rm e} |_\Sigma), \phi \right)_\Sigma
\end{align*}
for all $\phi \in H^{1/2} (\Sigma)$, which leads to the first identity in~\eqref{eq:conditionsNew}. In order to obtain the third equality in~\eqref{eq:conditionsNew} observe that the first equality in~\eqref{eq:conditionsNew} can be used to write~\eqref{eq:veryLongCalc} as
\begin{align}\label{eq:towards3}
\begin{split}
 0 & = \left(  \tfrac{1 + \overline \gamma / 2}{\beta} \eins_{\Sigma_\beta} f_{\rm i} |_\Sigma + \tfrac{\overline \gamma / 2 - 1}{\beta} \eins_{\Sigma_\beta} f_{\rm e} |_\Sigma, (1 - \tfrac{\overline{\gamma}}{2}) g_{\rm e} |_\Sigma - (1 + \tfrac{\overline{\gamma}}{2}) g_{\rm i} |_\Sigma \right)_\Sigma \\
 & \quad + \tfrac{1}{2} \left( (1 + \tfrac{\gamma}{2}) \partial_{\nu_{\rm i}} f_{\rm i} |_\Sigma - (1 + \tfrac{\gamma}{2}) \partial_{\nu_{\rm e}} f_{\rm e} |_\Sigma + \tfrac{\alpha}{2} (f_{\rm i} |_\Sigma + f_{\rm e} |_\Sigma), g_{\rm i} |_\Sigma \right)_\Sigma \\
 & \quad + \tfrac{1}{2} \left( (1 - \tfrac{\gamma}{2}) \partial_{\nu_{\rm e}} f_{\rm e} |_\Sigma - (1 - \tfrac{\gamma}{2}) \partial_{\nu_{\rm i}} f_{\rm i} |_\Sigma + \tfrac{\alpha}{2} (f_{\rm i} |_\Sigma + f_{\rm e} |_\Sigma), g_{\rm e} |_\Sigma \right)_\Sigma \\
 & \quad - \left( \tfrac{\alpha}{4} (f_{\rm i} |_\Sigma + f_{\rm e} |_\Sigma), g_{\rm i} |_\Sigma + g_{\rm e} |_\Sigma \right)_\Sigma \\
 & = \left(  \tfrac{1 + \overline \gamma / 2}{\beta} \eins_{\Sigma_\beta} f_{\rm i} |_\Sigma + \tfrac{\overline \gamma / 2 - 1}{\beta} \eins_{\Sigma_\beta} f_{\rm e} |_\Sigma, (1 - \tfrac{\overline{\gamma}}{2}) g_{\rm e} |_\Sigma - (1 + \tfrac{\overline{\gamma}}{2}) g_{\rm i} |_\Sigma \right)_\Sigma \\ 
 & \quad + \tfrac{1}{2} \left( (1 + \tfrac{\gamma}{2}) (\partial_{\nu_{\rm i}} f_{\rm i} |_\Sigma - \partial_{\nu_{\rm e}} f_{\rm e} |_\Sigma), g_{\rm i} |_\Sigma \right)_\Sigma \\
 & \quad + \tfrac{1}{2} \left( - (1 - \tfrac{\gamma}{2}) (\partial_{\nu_{\rm i}} f_{\rm i} |_\Sigma - \partial_{\nu_{\rm e}} f_{\rm e} |_\Sigma), g_{\rm e} |_\Sigma \right)_\Sigma \\
 & = \left(  \tfrac{1 + \overline \gamma / 2}{\beta} \eins_{\Sigma_\beta} f_{\rm i} |_\Sigma + \tfrac{\overline \gamma / 2 - 1}{\beta} \eins_{\Sigma_\beta} f_{\rm e} |_\Sigma, (1 - \tfrac{\overline{\gamma}}{2}) g_{\rm e} |_\Sigma - (1 + \tfrac{\overline{\gamma}}{2}) g_{\rm i} |_\Sigma \right)_\Sigma \\ 
 & \quad - \tfrac{1}{2} \left( \partial_{\nu_{\rm i}} f_{\rm i} |_\Sigma - \partial_{\nu_{\rm e}} f_{\rm e} |_\Sigma, (1 - \tfrac{\overline{\gamma}}{2}) g_{\rm e} |_\Sigma - (1 + \tfrac{\overline{\gamma}}{2}) g_{\rm i} |_\Sigma \right)_\Sigma
\end{split}
\end{align}
for all $g \in \dom h_\cA$. Moreover, for each $\phi \in H^{1/2} (\Sigma)$ with $\phi |_{\Sigma_0} = 0$ there exists $g \in \dom h_\cA$ with $(1 - \tfrac{\overline{\gamma}}{2}) g_{\rm e} |_\Sigma - (1 + \tfrac{\overline{\gamma}}{2}) g_{\rm i} |_\Sigma = \phi$, which can be obtained by choosing $g_j \in H^1 (\Omega_j)$, $j = \rm i, e$, such that $g_{\rm i} |_\Sigma = - \phi / 2$ and $g_{\rm e} |_\Sigma = \phi / 2$. Hence~\eqref{eq:towards3} is equivalent to
\begin{align*}
 0 = \left(  \tfrac{1 + \overline \gamma / 2}{\beta} \eins_{\Sigma_\beta} f_{\rm i} |_\Sigma + \tfrac{\overline \gamma / 2 - 1}{\beta} \eins_{\Sigma_\beta} f_{\rm e} |_\Sigma - \tfrac{1}{2} ( \partial_{\nu_{\rm i}} f_{\rm i} |_\Sigma - \partial_{\nu_{\rm e}} f_{\rm e} |_\Sigma), \phi \right)_\Sigma
\end{align*}
for all $\phi \in H^{1/2} (\Sigma)$ such that $\phi |_{\Sigma_0} = 0$, which is the third identity in~\eqref{eq:conditionsNew}. Thus $f \in \dom H$, that is, $- \Delta_\cA = H$. This completes the proof of the theorem.
\end{proof}

Let us mention two examples where the generalized interactions~\eqref{eq:conditions} reduce to situations which were studied before.

\begin{example}
The generalized interactions under consideration include, as special cases, the $\delta$-interaction on $\Sigma$ of strength $\alpha$ (setting $\beta = \gamma = 0$ identically) and the $\delta'$-interaction on $\Sigma$ of strength $\beta$ (setting $\alpha  = \gamma = 0$ identically). 
\end{example}

\begin{example}
Let $\alpha = \gamma = 0$ identically and let $\Sigma_\beta \neq \Sigma$. Then $- \Delta_\cA$ describes a $\delta'$-interaction of strength $\beta$ on the (possibly nonclosed) hypersurface $\Sigma_\beta$. In space dimension $n = 2$ and for special choices of the nonclosed curve $\Sigma_\beta$ spectral properties of this operator were studied recently in~\cite{JL15}. Similarly the interactions under consideration include $\delta$-interactions on non-closed hypersurfaces, which can be obtained by choosing $\gamma = \beta = 0$ identically and a function $\alpha : \Sigma \to \R$ being zero on a part of $\Sigma$ and nonzero on another part.
\end{example}

\section{Essential spectra and existence of bound states}

In this section we study the essential spectrum of $- \Delta_\cA$ in the cases of a compact hypersurface or a noncompact, asymptotically planar hypersurface. Moreover, for compact $\Sigma$ we derive conditions for the existence or absence of discrete, negative eigenvalues. 

In the following we write $\sigma (- \Delta_\cA)$, $\sigma_{\rm ess} (- \Delta_\cA)$ and $\rho (- \Delta_\cA)$ for the spectrum, essential spectrum and resolvent set of $- \Delta_\cA$, respectively, and denote by $N (- \Delta_\cA)$ the number of discrete eigenvalues below the bottom of the essential spectrum, counted with multiplicities.

\subsection{The case of a compact hypersurface}

Let us first consider the case of a compact hypersurface $\Sigma$. In the following theorem we denote by $- \Delta_{\rm free}$ the free Laplacian in $L^2 (\R^n)$ (which coincides with $- \Delta_\cA$ if $\cA$ is trivial).

\begin{theorem}\label{thm:compact}
Let $\Omega_{\rm i}$ be bounded, that is, $\Sigma$ is compact. Moreover, let Assumption~\ref{as:coefficients} be satisfied. Then the following assertions hold.
\begin{enumerate}
 \item The resolvent difference
 \begin{align*}
  (- \Delta_\cA - \lambda)^{-1} - (- \Delta_{\rm free} - \lambda)^{-1}, \quad \lambda \in \rho (- \Delta_\cA) \cap \rho (- \Delta_{\rm free}),
 \end{align*}
 is compact. In particular, $\sigma_{\rm ess} (- \Delta_\cA) = [0, \infty)$.
 \item The discrete spectrum $\sigma (- \Delta_\cA) \cap (- \infty, 0)$ is finite.
\end{enumerate}
\end{theorem}

\begin{proof}
(i) We proceed similar to the proof of Theorem~4.2 in~\cite{BEL14}. Let $\lambda \in \rho (- \Delta_\cA) \cap \rho (- \Delta_{\rm free})$ and let $f, g \in L^2 (\R^n)$. Define 
\begin{align*}
 W = (- \Delta_\cA - \lambda)^{-1} - (- \Delta_{\rm free} - \lambda)^{-1}
\end{align*}
and let $u = (- \Delta_\cA - \lambda)^{-1} f$ and $v = (- \Delta_{\rm free} - \overline \lambda)^{-1} g$. Then
\begin{align*}
 (W f, g)_{L^2 (\R^n)} & = (u, g)_{L^2 (\R^n)} - (f, v)_{L^2 (\R^n)} \\
 & = (u, - \Delta_{\rm free} v)_{L^2 (\R^n)} - (- \Delta_\cA u, v)_{L^2 (\R^n)}.
\end{align*}
Since $- \Delta_\cA u = - \Delta u_{\rm i} \oplus - \Delta u_{\rm e}$, cf.\ the proof of Theorem~\ref{thm:HA}, it follows from Green's identity
\begin{align}\label{eq:W}
\begin{split}
 (W f, g)_{L^2 (\R^n)} & = (\partial_{\nu_{\rm i}} u_{\rm i} |_\Sigma, v_{\rm i} |_\Sigma )_{\Sigma} - (u_{\rm i} |_\Sigma, \partial_{\nu_{\rm i}} v_{\rm i} |_\Sigma )_\Sigma \\
 & \quad + (\partial_{\nu_{\rm e}} u_{\rm e} |_\Sigma, v_{\rm e} |_\Sigma)_\Sigma - (u_{\rm e} |_\Sigma, \partial_{\nu_{\rm e}} v_{\rm e} |_\Sigma)_\Sigma \\
 & = (\partial_{\nu_{\rm i}} u_{\rm i} |_\Sigma + \partial_{\nu_{\rm e}} u_{\rm e} |_\Sigma, v |_\Sigma)_\Sigma - (u_{\rm i} |_\Sigma - u_{\rm e} |_\Sigma, \partial_{\nu_{\rm i}} v_{\rm i} |_\Sigma)_\Sigma;
\end{split}
\end{align}
in the last step we have used $v \in H^2 (\R^n)$, that is, $v_{\rm i} |_\Sigma = v_{\rm e} |_\Sigma$ and $\partial_{\nu_{\rm i}} v_{\rm i} |_\Sigma + \partial_{\nu_{\rm e}} v_{\rm e} |_\Sigma = 0$. Let us define operators $T_1, T_4 : L^2 (\R^n) \to H^{-1/2} (\Sigma)$ and $T_2, T_3 : L^2 (\R^n) \to H^{1/2} (\Sigma)$ by
\begin{align*}
 T_1 f & = \partial_{\nu_{\rm i}} ((- \Delta_\cA - \lambda)^{-1} f)_{\rm i} |_\Sigma + \partial_{\nu_{\rm e}} ((- \Delta_\cA - \lambda)^{-1} f)_{\rm e} |_\Sigma, \\
 T_2 g & = ((- \Delta_{\rm free} - \overline \lambda)^{-1} g)|_\Sigma, \\
 T_3 f & = ((- \Delta_\cA - \lambda)^{-1} f)_{\rm i} |_\Sigma - ((- \Delta_\cA - \lambda)^{-1} f)_{\rm e} |_\Sigma, \\
 T_4 g & = \partial_{\nu_{\rm i}} ((- \Delta_{\rm free} - \overline \lambda)^{-1} g)_{\rm i} |_\Sigma.
\end{align*}
Then $T_1, \dots, T_4$ are bounded, everywhere defined operators in the respective spaces, which follows from the continuity of the trace and the normal derivative from $H^1 (\Omega_j)$ to $H^{1/2} (\Sigma)$ and $H^{- 1/2} (\Sigma)$, respectively, $j = \rm i, e$. Moreover,~\eqref{eq:W} yields
\begin{align}\label{eq:W2}
 W = T_2^* T_1 - T_4^* T_3.
\end{align}
Note that $\ran (- \Delta_{\rm free} - \overline \lambda)^{-1} = H^2 (\R^n)$ implies $\ran T_2 \subset H^{1 - \eps} (\Sigma)$ for each $\eps \in (0, 1/2)$ and $\ran T_4 \subset L^2 (\Sigma)$; see~\cite[Section~3]{GM11} for the required properties of trace maps on Lipschitz domains. Since the embeddings of $H^{1 - \eps} (\Sigma)$ into $H^{1/2} (\Sigma)$ and of $L^2 (\Sigma)$ into $H^{- 1/2} (\Sigma)$ are compact it follows that $T_2$ and $T_4$ are compact. Together with~\eqref{eq:W2} this implies compactness of $W$, which completes the proof of assertion~(i).

(ii) Consider the quadratic form $\fa$ in $L^2 (\R^n)$ defined by
\begin{align*}
 \fa [f] & = \int_{\Omega_{\rm i}} |\nabla f_{\rm i}|^2 d x + \int_{\Omega_{\rm e}} |\nabla f_{\rm e}|^2 d x - \int_\Sigma \|\Theta_\cA (s)\| \big( |f_{\rm i} (s)|^2 + |f_{\rm e} (s)|^2 \big) d \sigma (s), \\
 \dom \fa & = H^1 (\Omega_{\rm i}) \oplus H^1 (\Omega_{\rm e}),
\end{align*}
where $\Theta_\cA$ is defined in~\eqref{eq:Theta} and $\| \cdot \|$ denotes the matrix norm induced by the Euclidean norm on $\C^2$. Due to the fact that $s \mapsto \|\Theta_\cA (s)\|$ is measurable and bounded on $\Sigma$, this form is densely defined, symmetric, semibounded from below and closed. The essential spectrum of the corresponding selfadjoint operator $A$ in $L^2 (\R^n)$ equals $[0, \infty)$, and its negative spectrum is finite, see~\cite[Theorem~6.9]{B62}. Moreover, we have $\dom h_\cA \subset \dom \fa$ and
\begin{align*}
  \fa [f] \leq h_\cA [f], \quad f \in \dom h_\cA,
\end{align*}
which implies $N (- \Delta_\cA) \leq N (A) < \infty$. This proves assertion~(ii) of the theorem.
\end{proof}

Let us investigate the existence of negative eigenvalues for $- \Delta_\cA$ if $\Sigma$ is compact. First we consider the case that $\beta (s) \neq 0$ for all $s \in \Sigma$. The following theorem is an extension of~\cite[Theorem~4.4]{BEL14}, where $\delta'$-interactions were considered.

\begin{theorem}\label{thm:BoundStates1}
Assume that $\Omega_{\rm i}$ is bounded, that is, $\Sigma$ is compact. Let Assumption~\ref{as:coefficients} be satisfied and assume that $\Sigma = \Sigma_\beta$, i.e., $\beta (s) \neq 0$ for all $s \in \Sigma$. If
\begin{align}\label{eq:intCondition}
 \int_\Sigma \bigg(\frac{|1 + \frac{\gamma}{2}|^2}{\beta} + \frac{\alpha}{4} \bigg) d \sigma > 0
\end{align}
then $N (- \Delta_\cA) > 0$ holds.
\end{theorem}

\begin{proof}
Let $h_\cA$ be the quadratic form corresponding to $- \Delta_\cA$. Since $\Sigma = \Sigma_\beta$ we have $\dom h_\cA = H^1 (\Omega_{\rm i}) \oplus H^1 (\Omega_{\rm e})$. In particular, the function $f = \eins_{\Omega_{\rm i}} \oplus 0$ belongs to $\dom h_\cA$. Moreover,
\begin{align*}
 h_\cA [f] & = - \int_\Sigma \bigg( \frac{|1 + \frac{\gamma}{2}|^2}{\beta} + \frac{\alpha}{4} \bigg) d \sigma < 0
\end{align*}
by~\eqref{eq:intCondition}. Thus $\min \sigma (- \Delta_\cA) < 0$. Since $\min \sigma_{\rm ess} (- \Delta_\cA) = 0$ by Theorem~\ref{thm:compact} it follows $N (- \Delta_\cA) > 0$.
\end{proof}

Theorem~\ref{thm:BoundStates1} leads to the following immediate corollary.

\begin{corollary}
Assume that $\Omega_{\rm i}$ is bounded, that is, $\Sigma$ is compact. Moreover, let Assumption~\ref{as:coefficients} be satisfied, let $\alpha (s) \geq 0$ and $\beta (s) > 0$ for all $s \in \Sigma$. If there exists a subset $\Sigma_1$ of $\Sigma$ of positive measure such that $\alpha (s) > 0$ or $\gamma (s) \neq -2$ for all $s \in \Sigma_1$ then $N (- \Delta_\cA) > 0$.
\end{corollary}

Let us now turn to the case $\beta = 0$ identically on $\Sigma$. In space dimension $n = 2$ the following holds.

\begin{theorem}\label{thm:BoundStates2}
Let $n = 2$ and let $\Sigma$ be compact. Let Assumption~\ref{as:coefficients} be satisfied with $\beta = 0$ identically on $\Sigma$. Moreover, let $\alpha (s) \geq \alpha_{\min} > 0$ for all $s \in \Sigma$ and let $\gamma \in \C$ be constant. Then $N (- \Delta_\cA) > 0$. 
\end{theorem}

\begin{proof}
Consider
\begin{align*}
 \widetilde \alpha = \frac{\alpha_{\min}}{\max \{ |1 + \gamma/2|^2, |1 - \gamma/2|^2 \}} > 0.
\end{align*}
Moreover, consider the quadratic form $h_{\widetilde \cA}$ in $L^2 (\R^2)$, defined as in~\eqref{eq:form} with
\begin{align*}
 \widetilde \cA = \begin{pmatrix} \widetilde \alpha & 0  \\ 0 & 0 \end{pmatrix},
\end{align*}
which corresponds to the Laplacian $- \Delta_{\widetilde \cA} = - \Delta_{\delta, \widetilde \alpha}$ in $L^2(\R^2)$ with a $\delta$-potential of strength $\widetilde \alpha$ on $\Sigma$. The negative (discrete) spectrum of $- \Delta_{\widetilde \cA}$ is nonempty~\cite{ET04,KL14}. In particular, there exists $g \in \dom h_{\widetilde \cA} = H^1 (\R^2)$ with $h_{\widetilde \cA} [g] < 0$. Let us assume for a moment $\gamma \neq \pm 2$. With $f = (1 + \overline{\gamma}/2)^{-1} g_{\rm i} \oplus (1 - \overline{\gamma}/2)^{-1} g_{\rm e}$ it follows $f \in \dom h_\cA$ and
\begin{align*}
 h_\cA [f] & = | 1 + \gamma/2|^{- 2} \int_{\Omega_{\rm i}} |\nabla g_{\rm i}|^2 d x + | 1 - \gamma/2|^{- 2} \int_{\Omega_{\rm e}} |\nabla g_{\rm e}|^2 d x \\
 & \quad - \int_\Sigma \frac{\alpha}{4} \big|(1 + \overline{\gamma}/2 )^{-1} g_{\rm i} |_\Sigma + (1 - \overline{\gamma}/2 )^{-1} g_{\rm e} |_\Sigma \big|^2 d \sigma \\
 & = | 1 + \gamma/2|^{- 2} \int_{\Omega_{\rm i}} |\nabla g_{\rm i}|^2 d x + | 1 - \gamma/2|^{- 2} \int_{\Omega_{\rm e}} |\nabla g_{\rm e}|^2 d x \\
 & \quad - \frac{\alpha_{\min}}{| 1 + \gamma/2|^2 | 1 - \gamma/2|^2} \int_\Sigma |g |_\Sigma|^2 d \sigma \\
 & \leq \frac{1}{\min \{ |1 + \gamma/2|^2, |1 - \gamma/2|^2 \}} \bigg\{ \int_{\Omega_{\rm i}} |\nabla g_{\rm i}|^2 d x + \int_{\Omega_{\rm e}} |\nabla g_{\rm e}|^2 d x \\
 & \qquad \qquad  - \frac{\alpha_{\min}}{\max \{ |1 + \gamma/2|^2, |1 - \gamma/2|^2 \}} \int_\Sigma \big| g|_\Sigma \big|^2 d \sigma \bigg\} \\
 & = \frac{1}{\min \{ |1 + \gamma/2|^2, |1 - \gamma/2|^2 \}} h_{\widetilde \cA} [g] < 0,
\end{align*}
where we have used $g_{\rm i} |_\Sigma = g_{\rm e} |_\Sigma$. Hence $N (- \Delta_\cA) > 0$. If $\gamma = 2$ or $\gamma = - 2$ we set $f = 0 \oplus g_{\rm e}$ or $f = g_{\rm i} \oplus 0$, respectively, and arrive at the same conclusion.
\end{proof}

In dimensions $n \geq 3$ the situation differs essentially. This is known for $\delta$-interactions~\cite{BEKS94,EF09} and the same holds true for the generalized interactions with $\beta = 0$ identically, as the following observation shows.

\begin{proposition}\label{prop:nonneg}
Let $\Sigma$ be compact and let Assumption~\ref{as:coefficients} be satisfied with $\beta = 0$ identically on $\Sigma$. Moreover, let  $0 \leq \alpha (s) \leq \alpha_{\max}$ for all $s \in \Sigma$ and let $\gamma \in \C$ be constant. Define 
\begin{align*}
 \widetilde \alpha = \frac{\alpha_{\max}}{\min \{ |1 + \gamma/2|^2, |1 - \gamma/2|^2 \}} \geq 0
\end{align*}
and let $- \Delta_{\delta, \widetilde \alpha}$ be the Schr\"odinger operator in $L^2 (\R^n)$ with $\delta$-potential of strength $\widetilde \alpha$ on $\Sigma$. If $N (- \Delta_{\delta, \widetilde \alpha}) = 0$ then $N (- \Delta_\cA) = 0$.
\end{proposition}

\begin{proof}
Note first that $- \Delta_{\delta, \widetilde \alpha} = - \Delta_{\widetilde \cA}$ with
\begin{align*}
 \widetilde \cA = \begin{pmatrix} \widetilde \alpha & 0  \\ 0 & 0 \end{pmatrix}.
\end{align*}
Assume that $N (- \Delta_{\delta, \widetilde \alpha}) = 0$. Let first $\gamma \neq \pm 2$. For $f \in \dom h_\cA$ we define $g = (1 + \overline{\gamma}/2) f_{\rm i} \oplus (1 - \overline{\gamma}/2) f_{\rm e}$. Then, clearly, $g \in H^1 (\R^n)$, $f = (1 + \overline{\gamma}/2)^{-1} g_{\rm i} \oplus (1 - \overline{\gamma}/2)^{-1} g_{\rm e}$, and we calculate similar to the proof of Theorem~\ref{thm:BoundStates2}
\begin{align*}
 h_\cA [f] & \geq \frac{1}{\max \{ |1 + \gamma/2|^2, |1 - \gamma/2|^2 \}} \bigg\{ \int_{\Omega_{\rm i}} |\nabla g_{\rm i}|^2 d x + \int_{\Omega_{\rm e}} |\nabla g_{\rm e}|^2 d x \\
 & \qquad \qquad  - \frac{\alpha_{\max}}{\min \{ |1 + \gamma/2|^2, |1 - \gamma/2|^2 \}} \int_\Sigma \big| g|_\Sigma \big|^2 d \sigma \bigg\} \\
 & = \frac{1}{\max \{ |1 + \gamma/2|^2, |1 - \gamma/2|^2 \}} h_{\widetilde \cA} [g] \geq 0.
\end{align*}
Since $f \in \dom h_\cA$ was chosen arbitrary it follows $N (- \Delta_\cA) = 0$. The cases $\gamma = \pm 2$ can be treated similarly, see the proof of Theorem~\ref{thm:BoundStates2}.
\end{proof}

The following example illustrates the possible absence of negative eigenvalues.

\begin{example}
Let $n = 3$ and let $\Sigma$ be a sphere of radius $R > 0$. Furthermore, let $\beta = 0$ identically, let $0 \leq \alpha (s) \leq \alpha_{\max}$ for all $s \in \Sigma$ and let $\gamma \in \C$ be constant. Define $\widetilde \alpha$ as in Proposition~\ref{prop:nonneg}. It was calculated in~\cite{AGS87} (see also~\cite[Example~4.1]{BEKS94}) that $N (- \Delta_{\delta, \widetilde \alpha}) = 0$ if and only if $\widetilde \alpha R \leq 1$. Thus it follows with the help of Proposition~\ref{prop:nonneg} that
\begin{align*}
 N (- \Delta_\cA) = 0 \quad \text{if} \quad \alpha_{\max} R \leq \min \{ |1 + \gamma/2|^2, |1 - \gamma/2|^2 \},
\end{align*}
that is, for sufficiently small $\alpha_{\max}$ and sufficiently large $|\gamma|$ (related to each other) no negative eigenvalues exist and $\sigma (- \Delta_\cA) = \sigma_{\rm ess} (- \Delta_\cA) = [0, \infty)$.
\end{example}

Let us mention that if $\beta$ is nontrivial on a part of $\Sigma$ and vanishes identically on another part then even in the two-dimensional case the operator $- \Delta_\cA$ may fail to exhibit bound states; cf.~\cite{JL15} for a $\delta'$-interaction on a nonclosed curve.

\subsection{The case of a noncompact, asymptotically planar hypersurface in $\R^3$}

In this paragraph we provide a result on the essential spectrum of $- \Delta_\cA$ in the case of a noncompact hypersurface $\Sigma$ in $\R^3$ which is asymptotically planar. For fixed numbers $\alpha, \beta \geq 0$ and $\gamma \in \C$ we define the matrix $\cA$ as in~\eqref{eq:A} and set
\begin{align}\label{eq:infPlane}
 m_\cA
 = \begin{cases} 
    - \frac{4 \alpha^2}{(4 + |\gamma|^2)^2}, & \text{if}~\beta = 0, \\
    - \frac{\left(4 + \det \cA + \sqrt{- 16 \alpha \beta + (4 + \det \cA)^2} \right)^2 }{16 \beta^2}, & \text{if}~\beta \neq 0.
   \end{cases}
\end{align}
As a preparation we formulate the following Proposition, where we use the notation
\begin{align*}
 \Theta_\cA = \begin{pmatrix}
           \frac{|1 + \frac{\gamma}{2}|^2}{\beta} + \frac{\alpha}{4} & \frac{(\frac{\overline{\gamma}}{2} - 1) (1 + \frac{\gamma}{2})}{\beta}  + \frac{\alpha}{4}\\
	   \frac{(\frac{\gamma}{2} - 1) (1 + \frac{\overline{\gamma}}{2})}{\beta} + \frac{\alpha}{4}  & \frac{|1 - \frac{\gamma}{2}|^2}{\beta}  + \frac{\alpha}{4}
          \end{pmatrix}
\end{align*}
if $\beta \neq 0$ and $\Theta_\cA = \binom{\alpha/4 \hspace{2mm} \alpha/4}{\alpha/4 \hspace{2mm} \alpha/4}$ if $\beta = 0$. 

\begin{proposition}\label{prop:schlimmesLemma}
Let $\alpha, \beta$ be nonnegative real numbers, let $\gamma \in \C$ and let $d > 0$. Then the quadratic form
\begin{align*}
 \eta_{\cA, d} [\psi] & = \int_{- d}^0 |\psi'|^2 d x + \int_0^d |\psi'|^2 d x - \left\langle \Theta_\cA \binom{\psi (0_-)}{\psi (0_+)}, \binom{\psi (0_-)}{\psi (0_+)} \right\rangle, \\
 \dom \eta_{\cA, d} & = \Big\{ \psi \in H^1 (-d, 0) \oplus H^1 (0, d) : ( 1 + \tfrac{\overline{\gamma}}{2} ) \psi (0_-) = ( 1 - \tfrac{\overline{\gamma}}{2}) \psi (0_+)~\text{if}~\beta = 0 \Big\},
\end{align*}
in $L^2 (-d, d)$ is semibounded from below by a constant $m_{\cA, d} \leq m_\cA$. Moreover, the following assertions hold.
\begin{enumerate}
 \item If $n_\tau$ is a family of real numbers with $n_\tau \to 1$ as $\tau \to \infty$ and
  \begin{align*}
   \widetilde \cA (\tau) = \begin{pmatrix} \frac{\alpha}{n_\tau} & \gamma \\ - \overline{\gamma} & n_\tau \beta \end{pmatrix}
  \end{align*}
  then $m_{\widetilde \cA (\tau), d} \to m_{\cA, d}$ as $\tau \to \infty$;
 \item $\lim_{d \to \infty} m_{\cA, d} = m_\cA$.
\end{enumerate}
\end{proposition}

The proof of Proposition~\ref{prop:schlimmesLemma} is longish and we shift it to Appendix~\ref{appendix}.

We come to the estimation of the essential spectrum. We make the following assumption. Here we denote by $B (0, r)$ the ball of radius $r$ in $\R^2$ centered at zero.

\begin{assumption}\label{as:asymptPlanar}
Assume that $\Sigma$ is a noncompact hypersurface described by a global Lipschitz parametrization $\phi : \R^2 \to \R^3$ with $\phi (\R^2) = \Sigma$ and there exists $\tau_0 > 0$ such that $\phi |_{\R^2 \setminus B (0, \tau_0)}$ is $C^2$-smooth and the Jacobian $(D \phi) (x')$ has rank two for all $x' \in \R^2 \setminus B (0, \tau_0)$. Assume furthermore that the following conditions are satisfied.
\begin{enumerate}
 \item[(a)] For each $d > 0$ there exists $\tau > \tau_0$ such that the mapping
\begin{align*}
 (\R^2 \setminus B (0, \tau)) \times (- d, d) \ni (x', x_3) \mapsto \phi (x') + x_3 \nu (x')
\end{align*}
is injective, where $\nu (x') = (\partial_1 \phi (x') \times \partial_2 \phi (x'))/{|\partial_1 \phi (x') \times \partial_2 \phi (x')|}$, $x' \in \R^2$, is a unit normal vector field of $\Sigma$.
 \item[(b)] The mean curvature $M$ and the Gau\ss{} curvature $K$ satisfy
\begin{align*}
 M (x') \to 0 \quad \text{and} \quad K (x') \to 0 \quad \text{as}~|x'| \to \infty.
\end{align*}
\end{enumerate}
\end{assumption}

Under these conditions the essential spectrum of $- \Delta_\cA$ can be estimated as follows. A variant of this theorem for $\delta$-interactions is contained in~\cite{EK03}; cf.\ also~\cite{DEK01} for a similar strategy of proof.

\begin{theorem}\label{thm:essSpecAsymPlanar}
Let $\Sigma$ satisfy Assumption~\ref{as:asymptPlanar}, and let $\alpha, \beta, \gamma$ be functions which satisfy Assumption~\ref{as:coefficients} and are constant outside a compact subset $\cK$ of $\Sigma$. Assume, additionally, that $\alpha (s) \geq 0$ and $\beta (s) \geq 0$ for all $s \in \Sigma$. Then
\begin{align}\label{eq:sigmaEss}
 \sigma_{\rm ess} (- \Delta_\cA) \subset [m_\cA, \infty)
\end{align}
holds, where $m_\cA$ is defined as in~\eqref{eq:infPlane} using the constant values of $\alpha, \beta$ and $\gamma$ in $\Sigma \setminus \cK$.
\end{theorem}

\begin{proof}
We carry out the proof for the case that the parametrization $\phi$ is globally $C^2$. The general case, where $\phi$ is allowed to be less regular inside a compact set, follows afterwards with a compact perturbation argument. We assume without loss of generality that the parametrization $\phi$ is chosen such that $\nu$ is the outer unit normal of $\Omega_{\rm i}$. Let $d > 0$ be sufficiently large such that the $d$-neighborhood $\Omega_ d = \Phi (\R^2 \times (-d, d))$ of $\Sigma$ described by the parametrization
\begin{align*}
 \Phi : \R^2 \times (-d, d) \to \R^3, \quad (x', x_3) \mapsto \phi (x') + x_3 \nu (x'), \quad x' \in \R^2, x_3 \in (-d, d),
\end{align*}
is a Lipschitz domain. Choose $\tau > \tau_0$ such that the restriction of $\Phi$ to
\begin{align*}
 D_{\tau, d} = \left\{ (x', x_3) \in \R^2 \times (-d, d) : |x'| > \tau \right\}
\end{align*}
is injective and such that the coefficients $\alpha, \beta$ and $\gamma$ are constants outside $\phi (B (0, \tau))$; note that these properties are then true also for each $\widetilde \tau > \tau$. We write $\Omega_{\tau, d}^{\rm ext} = \Phi (D_{\tau, d})$ and $\Omega_{\tau, d}^{\rm int} = \Omega_d \setminus \Omega_{\tau, d}^{\rm ext}$, and the latter is a bounded Lipschitz domain. Moreover, we set
\begin{align*}
 \Omega_d^{\rm ext} = \R^3 \setminus \overline{\Omega_d},
\end{align*}
which is an unbounded Lipschitz domain with a noncompact boundary. We denote by $- \Delta_{d, \rm N}^{\rm ext}$ the selfadjoint Neumann Laplacian in $L^2 (\Omega_d^{\rm ext})$. Furthermore, for $j = \rm int, ext$ we denote by $- \Delta_{\cA, d, \rm N}^{\tau, j}$ the selfadjoint Laplacian in $L^2 (\Omega_{\tau, d}^j)$ corresponding to the quadratic form
\begin{align*}
 h_{\cA, d, \rm N}^{\tau, j} [f] & = \int_{\Omega_{\tau, d, \rm i}^j} |\nabla f_{\rm i}|^2 d x + \int_{\Omega_{\tau, d, \rm e}^j} |\nabla f_{\rm e}|^2 d x - \int_{\Sigma_\tau^{j}} \left\langle \Theta_\cA \binom{f_{\rm i}}{f_{\rm e}}, \binom{f_{\rm i}}{f_{\rm e}} \right\rangle d \sigma, \\
 \dom h_{\cA, d, \rm N}^{\tau, j} & = \Big\{ f_{\rm i} \oplus f_{\rm e} \in H^1 (\Omega_{\tau, d, \rm i}^j) \oplus H^1 (\Omega_{\tau, d, \rm e}^j) : \\
& \qquad \qquad \qquad \big(1 + \tfrac{\overline{\gamma}}{2} \big) f_{\rm i} = \big(1 - \tfrac{\overline{\gamma}}{2} \big) f_{\rm e}~\text{on}~\Sigma_\tau^{j} \cap \Sigma_0 \Big\},
\end{align*}
where $\Omega_{\tau, d, k}^j = \Omega_k \cap \Omega_{\tau, d}^{j}$, $k = \rm i, e$, and $\Sigma_\tau^j := \Sigma \cap \Omega_{\tau, d}^j$. Then
\begin{align*}
 -\Delta_{d, \rm N}^{\rm ext} \oplus - \Delta_{\cA, d, \rm N}^{\tau, \rm int} \oplus - \Delta_{\cA, d, \rm N}^{\tau, \rm ext} \leq - \Delta_\cA
\end{align*}
and thus
\begin{align}\label{eq:essEst}
 \min \left\{ \min \sigma_{\rm ess} (- \Delta_{\cA, d, \rm N}^{\tau, \rm ext}), 0 \right\} \leq \min \sigma_{\rm ess} (- \Delta_\cA),
\end{align}
where we took into account that $- \Delta_{\cA, d, \rm N}^{\tau, \rm int}$ has a compact resolvent and that $- \Delta_{d, \rm N}^{\rm ext}$ has essential spectrum $[0, \infty)$. 

Let us estimate the quadratic form $h_{\cA, d, \rm N}^{\tau, \rm ext}$. Note that by assumption either $\Sigma_\tau^{\rm ext} \cap \Sigma_0 = \Sigma_\tau^{\rm ext}$ or $\Sigma_\tau^{\rm ext} \cap \Sigma_0 = \emptyset$. Note that the Jacobian $D \Phi$ of $\Phi$ satisfies
\begin{align*}
 (D \Phi)^\top D \Phi & = \begin{pmatrix}
                      |\partial_1 (\phi + x_3 \nu)|^2 & \langle \partial_1 (\phi + x_3 \nu), \partial_2 (\phi + x_3 \nu) \rangle & 0 \\
		      \langle \partial_2 (\phi + x_3 \nu), \partial_1 (\phi + x_3 \nu ) \rangle & |\partial_2 (\phi + x_3 \nu)|^2 & 0 \\
		      0 & 0 & 1
                     \end{pmatrix} \\
 & = \begin{pmatrix}( I - x_3 H )^\top \begin{pmatrix} \partial_1 \phi & \partial_2 \phi \end{pmatrix}^\top \begin{pmatrix} \partial_1 \phi & \partial_2 \phi \end{pmatrix} ( I - x_3 H ) & 0 \\ 0 & 1 \end{pmatrix},
\end{align*}
where the derivatives $\partial_1$ and $\partial_2$ refer to the variables $x' = (x_1, x_2)$ on which $\phi$ and $\nu$ depend and $H = (H_{i j})$ is the Weingarten map. Consequently, 
\begin{align*}
 (\det D \Phi)^2 & = \det [(D \Phi)^\top D \Phi] \\
 & = \big(1 - x_3 (H_{11} + H_{22}) + x_3^2 \det H \big)^2 \det \left(\begin{pmatrix} \partial_1 \phi & \partial_2 \phi \end{pmatrix}^\top \begin{pmatrix} \partial_1 \phi & \partial_2 \phi \end{pmatrix} \right) \\
 & = \big( 1 - 2 x_3 M + x_3^2 K \big)^2 |\partial_1 \phi \times \partial_2 \phi|^2.
\end{align*}
With these observations for $f \in \dom h_{\cA, d, \rm N}^{\tau, \rm ext}$ we find
\begin{align*}
 \int_{\Omega_{\tau, d, \rm i}^{\rm ext}} |\nabla f_{\rm i}|^2 d x & = \int_{D_{\tau, d} \cap \R^3_-} \big\langle \nabla (f_{\rm i} \circ \Phi), (D \Phi)^{- 1} (D \Phi)^{- \top} \nabla (f_{\rm i} \circ \Phi) \big\rangle |\det D \Phi| d x \\
 & \geq n_{\tau, d} \int_{D_{\tau, d} \cap \R^3_-} |\partial_3 (f_{\rm i} \circ \Phi) |^2 |\partial_1 \phi \times \partial_2 \phi| d x
\end{align*}
with $n_{\tau, d} = \inf_{x = (x', x_3) \in D_{\tau, d}} |1 - 2 x_3 M (x') + x_3^2 K (x')|$. Carrying out the analogous estimate for $\Omega_{\tau, d, \rm e}$ and $f_{\rm e}$ we arrive at
\begin{align*}
 h_{\cA, d, \rm N}^{\tau, \rm ext} [f] & \geq n_{\tau, d} \int_{\R^2 \setminus B (0, \tau)} \bigg( \int_{- d}^0 |\partial_3 (f_{\rm i} \circ \Phi) |^2 d x_3 + \int_0^d |\partial_3 (f_{\rm i} \circ \Phi) |^2 d x_3 \\
 & \qquad \qquad - \frac{1}{n_{\tau, d}} \left\langle \Theta_\cA \binom{f_{\rm i} \circ \phi}{f_{\rm e} \circ \phi}, \binom{f_{\rm i} \circ \phi}{f_{\rm e} \circ \phi} \right\rangle  \bigg) |\partial_1 \phi \times \partial_2 \phi| d x' \\
 & = n_{\tau, d} \int_{\R^2 \setminus B (0, \tau)} \eta_{\widetilde \cA (\tau, d), d} [(f \circ \Phi) (x', \cdot)] |\partial_1 \phi \times \partial_2 \phi| d x'
\end{align*}
for each $f \in \dom h_{\cA, d, \rm N}^{\tau, \rm ext}$, where 
\begin{align*}
 \widetilde \cA (\tau, d) = \begin{pmatrix}
                           \frac{\alpha}{n_{\tau, d}} & \gamma \\ - \overline{\gamma} & n_{\tau, d} \beta
                          \end{pmatrix}.
\end{align*}
Thus with $N_{\tau, d} = \sup_{x = (x', x_3) \in D_{\tau, d}} |1 - 2 x_3 M (x') + x_3^2 K (x')|$ we obtain from Proposition~\ref{prop:schlimmesLemma}
\begin{align*}
 h_{\cA, d, \rm N}^{\tau, \rm ext} [f] & \geq \frac{n_{\tau, d}}{N_{\tau, d}} m_{\widetilde A (\tau, d), d} \int_{D_{\tau, d}} |(f \circ \Phi)|^2 |1 - 2 x_3 M (x') + x_3^2 K (x')| |\partial_1 \phi \times \partial_2 \phi| d x \\
 & = \frac{n_{\tau, d}}{N_{\tau, d}} m_{\widetilde A (\tau, d), d} \int_{\Omega_{\tau, d}^{\rm ext}} |f|^2 d x,
\end{align*}
for each $f \in \dom h_{\cA, d, \rm N}^{\tau, \rm ext}$. It follows with the help of~\eqref{eq:essEst}
\begin{align}\label{eq:dtauest}
 \min \sigma_{\rm ess} (- \Delta_\cA) \geq \min \left\{ \min \sigma (- \Delta_{\cA, d, \rm N}^{\tau, \rm ext}), 0 \right\} \geq \frac{n_{\tau, d}}{N_{\tau, d}} m_{\widetilde A (\tau, d), d}.
\end{align}
Sending $\tau \to \infty$ yields $n_{\tau, d} \to 1$ and $N (\tau, d) \to 1$, and with the help of Proposition~\ref{prop:schlimmesLemma}~(i) we conclude from~\eqref{eq:dtauest}
\begin{align*}
 \min \sigma_{\rm ess} (- \Delta_\cA) \geq m_{\cA, d}.
\end{align*}
From this and Proposition~\ref{prop:schlimmesLemma}~(ii), sending $d \to \infty$ we obtain
\begin{align*}
 \min \sigma_{\rm ess} (- \Delta_\cA) \geq m_\cA
\end{align*}
and the assertion of the theorem follows.
\end{proof}

\begin{remark}\label{rem:plane}
In many cases we have equality in~\eqref{eq:sigmaEss}. For instance, if $\Sigma$ is a plane in $\R^3$ and $\alpha \geq 0, \beta \geq 0$ and $\gamma \in \C$ are constants then it can be shown by using the tensor product structure that $\sigma_{\rm ess} (- \Delta_\cA) = \sigma (- \Delta_\cA) = [m_\cA, \infty)$. Moreover, a local, compact deformation of the plane or a change of the coefficients $\alpha, \beta$ and $\gamma$ inside a compact set does not change the essential spectrum.
\end{remark}

Let us point out that Theorem~\ref{thm:essSpecAsymPlanar} can be combined with the operator inequalities obtained in Section~4 below in order to find the exact minimum of the essential spectrum; cf.\ Example~\ref{ex:cone}.

We remark that for noncompact $\Sigma$ the existence of bound states as well as the finiteness of the number of bound states depend strongly on geometric properties of $\Sigma$. See, e.g.,~\cite[Section~3]{E08} for a review in the case of $\delta$-interactions. Also an infinite number of bound states may occure. 

\section{Operator inequalities and spectral consequences}

In this section we prove inequalities between Laplacians with generalized interactions on a given hypersurface $\Sigma$. We are here back to the general situation, i.e., $\Sigma$ is the boundary of a Lipschitz domain in any dimension $n \geq 2$ and can be compact or unbounded. Our main focus is on instances of generalized interactions which allow an operator inequality against a $\delta$-interaction of an appropriate strength. This allows us, in particular, to derive spectral properties of generalized interactions from the corresponding well-studied properties of $\delta$--interactions. This strategy was applied to $\delta'$-interactions recently in~\cite{BEL14}. 

Recall that for two selfadjoint operators $H_1$ and $H_2$ in a Hilbert space which are semibounded below we write $H_1 \leq H_2$ if and only if
\begin{align}\label{eq:ordering}
 (H_1 - \lambda)^{-1} - (H_2 - \lambda)^{-1} \geq 0
\end{align}
holds for all $\lambda \leq \min \{\min \sigma (H_1), \min \sigma (H_2) \}$. Moreover, denoting the quadratic forms corresponding to $H_1$ and $H_2$ by $h_1$ and $h_2$, respectively,~\eqref{eq:ordering} is equivalent to $h_1 \leq h_2$, i.e., $\dom h_2 \subset \dom h_1$ and $h_1 [f] \leq h_2 [f]$ for all $f \in \dom h_2$. For further details see~\cite[Chapter~VI]{K95}. Let us emphasize some immediate spectral consequences of such an ordering, which follow from the min-max principle; cf.~\cite[Theorem~XIII.2]{RS78}. For $i = 1, 2$ we denote by $\lambda_k (H_i)$ the values defined by the min-max principle, i.e., the eigenvalues of $H_i$ below the bottom of the essential spectrum, counted with multiplicities, including the minimum of the essential spectrum if the number of eigenvalues below $\min \sigma_{\rm ess} (H_i)$ is finite. Moreover, as above we denote by $N (H_i)$ the number of eigenvalues of $H_i$ below $\min \sigma_{\rm ess} (H_i)$, counted with multiplicities.

\begin{proposition}\label{prop:ordering}
Let $H_1$ and $H_2$ be selfadjoint operators in a Hilbert space which are semibounded below and satisfy $H_1 \leq H_2$. Then the following assertions hold.
\begin{enumerate}
 \item $\min \sigma_{\rm ess} (H_1) \leq \min \sigma_{\rm ess} (H_2)$.
 \item $\lambda_k (H_1) \leq \lambda_k (H_2)$ for all $k \in \N$.
 \item If $\min \sigma_{\rm ess} (H_1) = \min \sigma_{\rm ess} (H_2)$ then
\begin{align*}
 N (H_2) \leq N (H_1).
\end{align*}
In particular, if $N (H_2) > 0$ then $N (H_1) > 0$.
\end{enumerate}
\end{proposition}

In this section we write again $- \Delta_{\delta, \widetilde \alpha}$ for the selfadjoint Laplacian subject to a $\delta$-interaction of strength $\widetilde \alpha$ in $L^2 (\R^n)$, i.e., $- \Delta_{\delta, \widetilde \alpha} = - \Delta_{\widetilde \cA}$ with
\begin{align*}
 \widetilde \cA = \begin{pmatrix} \widetilde \alpha & 0  \\ 0 & 0 \end{pmatrix},
\end{align*}
where $\widetilde \alpha : \Sigma \to \R$ is a measurable, bounded function.

We start with a class of interactions being a combination of $\delta$ and $\delta'$.

\begin{theorem}\label{thm:comparison1}
Let $\alpha, \beta : \Sigma \to \R$ be functions such that $\alpha$ and $1/\beta$ are bounded and measurable with $\alpha (s) \geq 0$ and $\beta (s) > 0$ for all $s \in \Sigma$. Let $- \Delta_\cA$ be the selfadjoint operator in Definition~\ref{def:HA} corresponding to
\begin{align}\label{eq:Acomp1}
 \cA = \begin{pmatrix} \alpha & 0 \\ 0 & \beta \end{pmatrix}.
\end{align}
Moreover, let $\widetilde \alpha : \Sigma \to \R$ be measurable and bounded. Then the following assertions hold.
\begin{enumerate}
 \item If $\widetilde \alpha (s) \leq \alpha (s)$ for all $s \in \Sigma$ then
 \begin{align*}
  - \Delta_\cA \leq - \Delta_{\delta, \widetilde \alpha}.
 \end{align*}
 \item If $\widetilde \alpha (s) \leq \frac{4}{\beta (s)}$ for all $s \in \Sigma$ then there exists a unitary operator $U$ in~$L^2 (\R^n)$ such that
 \begin{align*}
  U^* (- \Delta_\cA) U \leq - \Delta_{\delta, \widetilde \alpha}.
 \end{align*}
\end{enumerate}
\end{theorem}

\begin{proof}
If $\widetilde \alpha (s) \leq \alpha (s)$ holds for all $s \in \Sigma$ then we have 
\begin{align*}
 \dom h_{\widetilde A} = H^1 (\R^n) \subset H^1 (\Omega_{\rm i}) \oplus H^1 (\Omega_{\rm e}) = \dom h_\cA
\end{align*}
and for each $f \in \dom h_{\widetilde A}$
\begin{align*}
 h_\cA [f] & = \int_{\Omega_{\rm i}} |\nabla f_{\rm i}|^2 d x + \int_{\Omega_{\rm e}} |\nabla f_{\rm e}|^2 d x - \int_\Sigma \bigg\langle \binom{(\frac{1}{\beta} + \frac{\alpha}{4}) f_{\rm i} + (- \frac{1}{\beta} + \frac{\alpha}{4}) f_{\rm e}}{(- \frac{1}{\beta} + \frac{\alpha}{4}) f_{\rm i} + (\frac{1}{\beta} + \frac{\alpha}{4}) f_{\rm e}}, \binom{f_{\rm i}}{f_{\rm e}} \bigg\rangle d \sigma \\
 & = \int_{\Omega_{\rm i}} |\nabla f_{\rm i}|^2 d x + \int_{\Omega_{\rm e}} |\nabla f_{\rm e}|^2 d x - \int_\Sigma \alpha |f|^2 d \sigma \\
 & \leq \int_{\Omega_{\rm i}} |\nabla f_{\rm i}|^2 d x + \int_{\Omega_{\rm e}} |\nabla f_{\rm e}|^2 d x - \int_\Sigma \widetilde \alpha |f|^2 d \sigma \\
 & = h_{\widetilde \cA} [f]
\end{align*}
holds since $f_{\rm i} |_\Sigma = f_{\rm e} |_\Sigma$. Thus $h_{\cA} \leq h_{\widetilde \cA}$, which implies the assertion~(i).

If $\widetilde \alpha (s) \leq \frac{4}{\beta (s)}$ for all $s \in \Sigma$ we define an operator
\begin{align*}
 U : L^2 (\R^n) \to L^2 (\R^n), \quad f = f_{\rm i} \oplus f_{\rm e} \mapsto - f_{\rm i} \oplus f_{\rm e}.
\end{align*}
Clearly, $U$ is unitary. Thus the quadratic form 
\begin{align*}
 \widetilde h_\cA [f, g] = h_\cA [U f, U g], \quad \dom \widetilde h_\cA = U^* \dom h_\cA = H^1 (\Omega_{\rm i}) \oplus H^1 (\Omega_{\rm e}),
\end{align*}
is densely defined, semibounded below and closed and the corresponding selfadjoint operator in $L^2 (\R^n)$ equals $U^* (- \Delta_\cA) U$. Moreover, $\dom h_{\widetilde \cA} \subset \dom \widetilde h_\cA$ and for all $f \in \dom h_{\widetilde \cA}$ we have
\begin{align*}
 \widetilde h_\cA [f] & = \int_{\Omega_{\rm i}} |\nabla f_{\rm i}|^2 d x + \int_{\Omega_{\rm e}} |\nabla f_{\rm e}|^2 d x \\
 & \quad - \int_\Sigma \bigg\langle \binom{- (\frac{1}{\beta} + \frac{\alpha}{4}) f_{\rm i} + (- \frac{1}{\beta} + \frac{\alpha}{4}) f_{\rm e}}{- (- \frac{1}{\beta} + \frac{\alpha}{4}) f_{\rm i} + (\frac{1}{\beta} + \frac{\alpha}{4}) f_{\rm e}}, \binom{- f_{\rm i}}{f_{\rm e}} \bigg\rangle d \sigma \\
 & = \int_{\Omega_{\rm i}} |\nabla f_{\rm i}|^2 d x + \int_{\Omega_{\rm e}} |\nabla f_{\rm e}|^2 d x - \int_\Sigma \tfrac{4}{\beta} |f|^2 d \sigma \\
 & \leq h_{\widetilde \cA} [f]
\end{align*}
since $f_{\rm i} |_\Sigma = f_{\rm e} |_\Sigma$. From this assertion~(ii) follows.
\end{proof}

For $\alpha = 0$ the assertion of Theorem~\ref{thm:comparison1}~(ii) is in accordance with~\cite[Theorem~3.6]{BEL14}.

Let us next consider a class of generalized interactions which is closer to $\delta'$ and involves a nontrivial coefficient $\gamma$.

\begin{theorem}\label{thm:comparison2}
Let $\beta : \Sigma \to \R$ be a function with $\beta (s) > 0$ for all $s \in \Sigma$ such that $1/\beta$ is measurable and bounded, and let $\gamma \in i \R$ be a constant. Let $- \Delta_\cA$ be the selfadjoint operator in Definition~\ref{def:HA} corresponding to
\begin{align}\label{eq:Acomp2}
 \cA = \begin{pmatrix} 0 & \gamma \\ - \overline \gamma & \beta \end{pmatrix}.
\end{align}
Moreover, let $\widetilde \alpha : \Sigma \to \R$ be measurable and bounded. Assume that
\begin{align}\label{eq:estimate2}
 \widetilde \alpha (s) \leq \frac{4 + |\gamma|^2}{\beta (s)}
\end{align}
holds for all $s \in \Sigma$. Then there exists a unitary operator $U$ in $L^2 (\R^n)$ such that
\begin{align*}
 U^* (- \Delta_\cA) U \leq - \Delta_{\delta, \widetilde \alpha}.
\end{align*}
\end{theorem}

\begin{proof}
Note first that $r := |1 + \frac{\gamma}{2}| = |1 - \frac{\gamma}{2}| = (1 + \frac{|\gamma|^2}{4})^{1/2}$ is nonzero and define
\begin{align*}
 U : L^2 (\R^n) \to L^2 (\R^n), \quad f = f_{\rm i} \oplus f_{\rm e} \mapsto - \frac{1 + \frac{\gamma}{2}}{r} f_{\rm i} \oplus \frac{1 - \frac{\gamma}{2}}{r} f_{\rm e}.
\end{align*}
Then $U$ is a unitary operator and the quadratic form
\begin{align*}
 \widetilde h_\cA [f, g] = h_\cA [U f, U g], \quad \dom \widetilde h_\cA = U^* \dom h_\cA = H^1 (\Omega_{\rm i}) \oplus H^1 (\Omega_{\rm e}),
\end{align*}
is densely defined, semibounded below and closed and corresponds to the selfadjoint operator $U^* (- \Delta_\cA) U$. Moreover, $\dom h_{\widetilde \cA} \subset \dom \widetilde h_\cA$ and for each $f \in \dom h_{\widetilde \cA} = H^1 (\R^n)$ we have
\begin{align*}
 \widetilde h_\cA [f] & = \int_{\Omega_{\rm i}} |\nabla f_{\rm i}|^2 d x + \int_{\Omega_{\rm e}} |\nabla f_{\rm e}|^2 d x \\
 & \quad - \int_\Sigma \left\langle \begin{pmatrix} \frac{- |1 + \frac{\gamma}{2}| (1 + \frac{\gamma}{2})}{\beta} f_{\rm i} - \frac{|1 - \frac{\gamma}{2}| (1 + \frac{\gamma}{2})}{\beta} f_{\rm e} \\ \frac{- |1 + \frac{\gamma}{2}|(\frac{\gamma}{2} - 1)}{\beta} f_{\rm i} + \frac{|1 - \frac{\gamma}{2}| (1 - \frac{\gamma}{2})}{\beta} f_{\rm e} \end{pmatrix}, \begin{pmatrix} - \frac{1 + \frac{\gamma}{2}}{r} f_{\rm i} \\ \frac{1 - \frac{\gamma}{2}}{r} f_{\rm e} \end{pmatrix} \right\rangle d \sigma \\
 & = \int_{\Omega_{\rm i}} |\nabla f_{\rm i}|^2 d x + \int_{\Omega_{\rm e}} |\nabla f_{\rm e}|^2 d x - \int_\Sigma \left( 2 \frac{|1 + \tfrac{\gamma}{2}|^2}{\beta} + 2 \frac{|1 - \tfrac{\gamma}{2}|^2}{\beta} \right) |f|^2 d \sigma \\
& = \int_{\Omega_{\rm i}} |\nabla f_{\rm i}|^2 d x + \int_{\Omega_{\rm e}} |\nabla f_{\rm e}|^2 d x - \int_\Sigma \frac{4 + |\gamma|^2}{\beta} |f|^2 d \sigma \\
& \leq h_{\widetilde \cA} [f]
\end{align*}
as $f_{\rm i} |_\Sigma = f_{\rm e} |_\Sigma$. Thus $\widetilde h_\cA \leq h_{\widetilde \cA}$ and the assertion of the theorem follows.
\end{proof}

Let us consider a further class of generalized interactions, denoted as generalized interactions of $\delta$-type in the one-dimensional case in~\cite{EF06}.

\begin{theorem}\label{thm:comparison3}
Let $\alpha : \Sigma \to \R$ be a bounded, measurable function with $\alpha (s) \geq 0$ for all $s \in \Sigma$ and let $\gamma \in i \R$ be a constant. Let $- \Delta_\cA$ be the selfadjoint operator in Definition~\ref{def:HA} corresponding to
\begin{align}\label{eq:Acomp3}
 \cA = \begin{pmatrix} \alpha & \gamma \\ - \overline \gamma & 0 \end{pmatrix}.
\end{align}
Moreover, let $\widetilde \alpha : \Sigma \to \R$ be measurable and bounded and assume that
\begin{align*}
 \widetilde \alpha (s) \leq \frac{\alpha (s)}{|1 + \frac{\gamma}{2}|^2}
\end{align*}
holds for all $s \in \Sigma$. Then there exists a unitary operator $U$ in $L^2 (\R^n)$ such that
\begin{align*}
 U^* (- \Delta_\cA) U \leq - \Delta_{\delta, \widetilde \alpha}.
\end{align*}
\end{theorem}

\begin{proof}
Let us set $r := |1 + \frac{\gamma}{2}| = |1 - \frac{\gamma}{2}| > 0$. Define
\begin{align*}
 U : L^2 (\R^n) \to L^2 (\R^n), \quad f = f_{\rm i} \oplus f_{\rm e} \mapsto \frac{1 + \frac{\gamma}{2}}{r} f_{\rm i} \oplus \frac{1 - \frac{\gamma}{2}}{r} f_{\rm e}.
\end{align*}
Then $U$ is a unitary operator. Let us define a sesquilinear form $\widetilde h_{\cA}$ in $L^2 (\R^n)$ by
\begin{align*}
 \widetilde h_{\cA} [f, g] = h_{\cA} [U f, U g], \quad \dom \widetilde h_{\cA} = U^* \dom h_{\cA}.
\end{align*}
Then $\widetilde h_{\cA}$ is densely defined, symmetric, semibounded below and closed and the corresponding selfadjoint operator in~$L^2 (\R^n)$ equals~$U^* H_{\cA} U$. Thus, in order to verify the assertion of the theorem we have to show $\dom h_{\widetilde \cA} \subset \dom \widetilde h_{\cA}$ and
\begin{align*}
 \widetilde h_{\cA} [f] \leq h_{\widetilde \cA} [f], \quad f \in \dom h_{\widetilde \cA}.
\end{align*}
For the domain inclusion let $f = f_{\rm i} \oplus f_{\rm e} \in \dom h_{\widetilde \cA} = H^1 (\R^n)$. Then
\begin{align*}
 \Big( 1 + \frac{\overline \gamma}{2} \Big) (U f)_{\rm i} |_\Sigma & = \Big( 1 + \frac{\overline \gamma}{2} \Big) \frac{1 + \frac{\gamma}{2}}{r} f_{\rm i} |_\Sigma = r f_{\rm i} |_\Sigma = r f_{\rm e} |_\Sigma \\
 & = \Big( 1 - \frac{\overline \gamma}{2} \Big) \frac{1 - \frac{\gamma}{2}}{r} f_{\rm e} |_\Sigma = \Big( 1 - \frac{\overline \gamma}{2} \Big) (U f)_{\rm e} |_\Sigma,
\end{align*}
thus $U f \in \dom h_{\cA}$, that is, $f \in \dom \widetilde h_{\cA}$. Moreover, for $f \in \dom h_{\widetilde \cA} = H^1 ( \R^n)$ we have
\begin{align*}
 \widetilde h_{\cA} [f] & = h_{\cA} \Big[ \tfrac{1 + \frac{\gamma}{2}}{r} f_{\rm i} \oplus \tfrac{1 - \frac{\gamma}{2}}{r} f_{\rm e} \Big] \\
 & = \int_{\Omega_{\rm i}} |\nabla f_{\rm i}|^2 d x + \int_{\Omega_{\rm e}} |\nabla f_{\rm e}|^2 d x - \int_\Sigma \frac{\alpha}{4} \Big| \tfrac{1 + \frac{\gamma}{2}}{r} f_{\rm i} + \tfrac{1 - \frac{\gamma}{2}}{r} f_{\rm e} \Big|^2 d \sigma \\
 & = \int_{\Omega_{\rm i}} |\nabla f_{\rm i}|^2 d x + \int_{\Omega_{\rm e}} |\nabla f_{\rm e}|^2 d x - \int_\Sigma \frac{\alpha}{r^2} | f |^2 d \sigma \\
& \leq h_{\widetilde \cA} [f],
\end{align*}
where we have used $f_{\rm i} |_\Sigma = f_{\rm e} |_\Sigma$. Hence $\widetilde h_{\cA} \leq h_{\widetilde \cA}$, which completes the proof.
\end{proof}

In the following corollary we collect spectral implications of the previous theorems; cf.\ Proposition~\ref{prop:ordering}. 

\begin{corollary}\label{cor:spec1}
Let either the assumptions of Theorem~\ref{thm:comparison1} or of Theorem~\ref{thm:comparison2} or of Theorem~\ref{thm:comparison3} be satisfied. Then the following assertions hold.
\begin{enumerate}
 \item $\min \sigma_{\rm ess} (- \Delta_\cA) \leq \min \sigma_{\rm ess} (- \Delta_{\delta, \widetilde \alpha})$.
 \item $\lambda_k (- \Delta_\cA) \leq \lambda_k (- \Delta_{\delta, \widetilde \alpha})$ for all $k \in \N$.
 \item If $\min \sigma_{\rm ess} (- \Delta_\cA) = \min \sigma_{\rm ess} (- \Delta_{\delta, \widetilde \alpha})$ then
\begin{align*}
 N (- \Delta_{\delta, \widetilde \alpha}) \leq N (- \Delta_\cA).
\end{align*}
In particular, if $N (- \Delta_{\delta, \widetilde \alpha}) > 0$ then $N (- \Delta_\cA) > 0$.
\end{enumerate}
\end{corollary}

We remark that for eigenvalues below the bottom of the essential spectrum the inequality in item~(ii) of this corollary may be strict in certain cases; cf.~\cite{LR15} for inequalities between $\delta$- and $\delta'$-eigenvalues. 

Note that the condition $\min \sigma_{\rm ess} (- \Delta_\cA) = \min \sigma_{\rm ess} (- \Delta_{\delta, \widetilde \alpha})$ is safisfied automatically if for instance $\Sigma$ is compact or if $\Sigma$ is a plane; cf.\ Remark~\ref{rem:plane}. In the following we provide another example of generalized interactions on a noncompact, asymptotically planar hypersurface where this equality holds as well. The example illustrates at the same time how the assertions of Corollary~\ref{cor:spec1} can be applied.

\begin{example}\label{ex:cone}
Let $\Sigma$ be a conical surface in $\R^3$, i.e., 
\begin{align*}
 \Sigma = \left\{ (x', x_3) \in \R^3 : x_3 = \cot (\theta) |x'| \right\}
\end{align*}
for some $\theta \in (0, \pi/2)$. Then $\Sigma$ is asymptotically planar, i.e., it can be parametrized in such a way that Assumption~\ref{as:asymptPlanar} is satisfied. By~\cite{BEL14-2} for a constant $\widetilde \alpha > 0$ the essential spectrum of $- \Delta_{\delta, \widetilde \alpha}$ equals
\begin{align}\label{eq:sigmaEssCone}
 \sigma_{\rm ess} (- \Delta_{\delta, \widetilde \alpha}) = \big[- \alpha^2 / 4, \infty \big). 
\end{align}
We consider three different types of generalized interactions on $\Sigma$ according to Theorem~\ref{thm:comparison1}--Theorem~\ref{thm:comparison3}.

(i) Let $\alpha \geq 0$ and $\beta > 0$ be constants such that $\alpha \leq 4 / \beta$ and let $\cA$ be given in~\eqref{eq:Acomp1}. Moreover, let $\widetilde \alpha = 4 / \beta$. Then Theorem~\ref{thm:essSpecAsymPlanar}, Theorem~\ref{thm:comparison1} and~\eqref{eq:sigmaEssCone} yield
\begin{align*}
 - \frac{\widetilde \alpha^2}{4} = - \frac{4}{\beta^2} = m_\cA \leq \min \sigma_{\rm ess} (- \Delta_\cA) \leq  \min \sigma_{\rm ess} (- \Delta_{\delta, \widetilde \alpha}) = - \frac{\widetilde \alpha^2}{4},
\end{align*}
where $m_\cA$ is given in~\eqref{eq:infPlane}. This implies equality; in particular, $\min \sigma_{\rm ess} (- \Delta_\cA) = \min \sigma_{\rm ess} (- \Delta_{\delta, \widetilde \alpha})$. 

(ii) Let $\beta > 0$ and $\gamma \in i \R$ be constants and let $\cA$ be given in~\eqref{eq:Acomp2}. Moreover, let $\widetilde \alpha = (4 + |\gamma|^2)/\beta$. Then Theorem~\ref{thm:essSpecAsymPlanar}, Theorem~\ref{thm:comparison2} and~\eqref{eq:sigmaEssCone} imply
\begin{align*}
 - \frac{\widetilde \alpha^2}{4} = - \frac{(4 + |\gamma|^2)^2}{4 \beta^2} = m_\cA \leq \min \sigma_{\rm ess} (- \Delta_\cA) \leq  \min \sigma_{\rm ess} (- \Delta_{\delta, \widetilde \alpha}) = - \frac{\widetilde \alpha^2}{4},
\end{align*}
which implies again $\min \sigma_{\rm ess} (- \Delta_\cA) = \min \sigma_{\rm ess} (- \Delta_{\delta, \widetilde \alpha})$. 

(iii) Let $\alpha > 0$ and $\gamma \in i \R$ be constants and let $\cA$ be given in~\eqref{eq:Acomp3}. Moreover, let $\widetilde \alpha = \alpha/|1 + \frac{\gamma}{2}|^2$. Then with Theorem~\ref{thm:essSpecAsymPlanar}, Theorem~\ref{thm:comparison3} and~\eqref{eq:sigmaEssCone} we get
\begin{align*}
 - \frac{\widetilde \alpha^2}{4} = - \frac{4 \alpha^2}{(4 + |\gamma|^2)^2} = m_\cA \leq \min \sigma_{\rm ess} (- \Delta_\cA) \leq  \min \sigma_{\rm ess} (- \Delta_{\delta, \widetilde \alpha}) = - \frac{\widetilde \alpha^2}{4},
\end{align*}
and thus again $\min \sigma_{\rm ess} (- \Delta_\cA) = \min \sigma_{\rm ess} (- \Delta_{\delta, \widetilde \alpha})$.

We remark that the same reasoning applies to any asymptotically planar hypersurface for which~\eqref{eq:sigmaEssCone} is known.

In each of the cases (i)--(iii) it follows from Corollary~\ref{cor:spec1} in combination with~\cite[Theorem~3.2]{BEL14-2} that
\begin{align*}
 N (- \Delta_\cA) = \infty
\end{align*}
and that
\begin{align*}
 \lambda_k (- \Delta_\cA) \leq \lambda_k (- \Delta_{\delta, \widetilde \alpha}) < - \frac{\widetilde \alpha^2}{4}
\end{align*}
holds for all $k \in \N$.
\end{example}

In the following we provide one more operator inequality. We show that a certain class of generalized interactions admits an estimate from below against the $\delta'$-operator of an appropriate strength. This class of interactions, where $\beta = 0$ and $\Real \gamma$ may be nontrivial, is called the intermediate class in~\cite{EF06}. For a function $\widetilde \beta : \Sigma \to \R$ such that $1 / \widetilde \beta$ is measurable and bounded we denote by $- \Delta_{\delta', \widetilde \beta}$ the selfadjoint operator in $L^2 (\R^n)$ corresponding to a $\delta'$-interaction of strength $\widetilde \beta$, i.e., $- \Delta_{\delta', \widetilde \beta} = H_{\widetilde \cA}$ with
\begin{align*}
 \widetilde \cA = \begin{pmatrix} 0 & 0 \\ 0 & \widetilde \beta \end{pmatrix}.
\end{align*} 
With this notation the following theorem holds.

\begin{theorem}\label{thm:comparison4}
Let $\alpha : \Sigma \to \R$ be a bounded, measurable function with $\alpha (s) \geq 0$ for all $s \in \Sigma$ and let $\gamma : \Sigma \to \C$ be measurable and bounded. Let $- \Delta_\cA$ be the selfadjoint operator in Definition~\ref{def:HA} corresponding to
\begin{align*}
 \cA = \begin{pmatrix} \alpha & \gamma \\ - \overline \gamma & 0 \end{pmatrix}.
\end{align*}
Moreover, let $\widetilde \beta : \Sigma \to \R$ be such that $1/\widetilde \beta$ is measurable and bounded and assume
\begin{align*}
 \alpha (s) \leq \frac{4}{\widetilde \beta (s)}
\end{align*}
for all $s \in \Sigma$. Then there exists a unitary operator $U$ in $L^2 (\R^n)$ such that
\begin{align*}
 U^* (- \Delta_{\delta', \widetilde \beta}) U \leq - \Delta_\cA.
\end{align*}
\end{theorem}

\begin{proof}
Consider the unitary operator
\begin{align*}
 U : L^2 (\R^n) \to L^2 (\R^n), \quad f = f_{\rm i} \oplus f_{\rm e} \mapsto f_{\rm i} \oplus - f_{\rm e}.
\end{align*}
Let us define a sesquilinear form $\widetilde h_{\widetilde \cA}$ in $L^2 (\R^n)$ by
\begin{align*}
 \widetilde h_{\widetilde \cA} [f, g] = h_{\widetilde \cA} [U f, U g], \quad \dom \widetilde h_{\widetilde \cA} = U^* \dom h_{\widetilde \cA} = H^1 (\Omega_{\rm i}) \oplus H^1 (\Omega_{\rm e}).
\end{align*}
Then the inclusion $\dom h_{\cA} \subset \dom \widetilde h_{\widetilde \cA}$ is obvious and for $f \in \dom h_{\cA}$ we have
\begin{align*}
 \widetilde h_{\widetilde \cA} [f] & = \int_{\Omega_{\rm i}} |\nabla f_{\rm i}|^2 d x + \int_{\Omega_{\rm e}} |\nabla f_{\rm e}|^2 d x - \int_\Sigma \frac{1}{\widetilde \beta} |f_{\rm i} + f_{\rm e} |^2 d \sigma \\
 & \leq \int_{\Omega_{\rm i}} |\nabla f_{\rm i}|^2 d x + \int_{\Omega_{\rm e}} |\nabla f_{\rm e}|^2 d x - \int_\Sigma \frac{\alpha}{4} | f_{\rm i} + f_{\rm e} |^2 d \sigma \\
& = h_{\cA} [f].
\end{align*}
From this the claim follows.
\end{proof}

We immediately obtain the following corollary on the spectrum. 

\begin{corollary}
Let the assumptions of Theorem~\ref{thm:comparison4} be satisfied. Then the following assertions hold.
\begin{enumerate}
 \item $\min \sigma_{\rm ess} (- \Delta_{\delta', \widetilde \beta}) \leq \min \sigma_{\rm ess} (- \Delta_\cA)$.
 \item $\lambda_k (- \Delta_{\delta', \widetilde \beta}) \leq \lambda_k (- \Delta_\cA)$ for all $k \in \N$.
 \item If $\min \sigma_{\rm ess} (- \Delta_\cA) = \min \sigma_{\rm ess} (- \Delta_{\delta', \widetilde \beta})$ then
\begin{align*}
 N (- \Delta_\cA) \leq N (- \Delta_{\delta', \widetilde \beta}).
\end{align*}
\end{enumerate}
\end{corollary}

\appendix

\section{Proof of Proposition~\ref{prop:schlimmesLemma}}\label{appendix}

In this appendix we provide a proof for Proposition~\ref{prop:schlimmesLemma}.

First, it is not difficult to verify that the form $\eta_{\cA, d}$ is semibounded from below; cf.~the proof of Lemma~\ref{lem:form}. Moreover, note that $\eta_{\cA}$ is also closed and densely defined and the selfadjoint operator in $L^2 (- d, d)$ corresponding to $\eta_{\cA, d}$ is given by
\begin{align*}
 H_{\cA, d} \psi = - \psi_-'' \oplus - \psi''_+, \qquad \psi = \psi_- \oplus \psi_+ \in \dom H_{\cA, d},
\end{align*}
where $\dom H_{\cA, d}$ consists of all $\psi = \psi_- \oplus \psi_+ \in H^2 (- d, 0) \oplus H^2 (0, d)$ which satisfy the conditions
\begin{align}\label{eq:GPI}
\begin{split}
 \psi' (0_-) - \psi' (0_+) & = \frac{\alpha}{2} \big(\psi (0_-) + \psi (0_+) \big) + \frac{\gamma}{2} \big(\psi' (0_-) + \psi' (0_+) \big), \\
 \psi (0_-) - \psi (0_+) & = - \frac{\overline{\gamma}}{2} \big(\psi (0_-) + \psi (0_+) \big) + \frac{\beta}{2} \big(\psi' (0_-) + \psi' (0_+) \big), \\
 \psi' (-d) & = 0, \\
 \psi' (d) & = 0.
\end{split}
\end{align}
The spectrum of $H_{\cA, d}$ is purely discrete; in particular, the infimum of the spectrum is given by an eigenvalue, which is nonpositive, as we will see. In the case $\alpha = \beta = 0$ clearly $\eta_{\cA, d}$ is nonnegative and $\eta_{\cA, d} [\psi] = 0$ holds for an appropriate normed, piecewise constant function, so that all statements of the lemma follow immediately. Therefore in the following we assume that $\alpha > 0$ or $\beta > 0$. Observe that $\lambda = - k^2$ with $k > 0$ is an eigenvalue of $H_{\cA, d}$ if and only if
\begin{align}\label{eq:spectralIdentity}
 g (k) = h (k) j (k),
\end{align}
where
\begin{align*}
 g (k) =  |\gamma|^2 k (1 - e^{4 k d}), \qquad h (k) = (- 2 \alpha - 4 k) e^{- 2 k d} - 2 \alpha + 4 k
\end{align*}
and
\begin{align*}
 j (k) = e^{2 k d} \Big(1 + \frac{\beta}{2} k \Big) + e^{4 k d} \Big( 1 - \frac{\beta}{2} k \Big).
\end{align*}
Indeed, each $\psi = \psi_- \oplus \psi_+ \in \ker (H_{\cA, d} - \lambda)$ satisfies
\begin{align*}
 \psi_- (x) = A e^{- k x} + B e^{k x} \quad \text{and} \quad \psi_+ (x) = C e^{- k x} + D e^{k x}
\end{align*}
and the conditions~\eqref{eq:GPI} turn into
\begin{align*}
 - A k + B k + C k - D k & = \frac{\alpha}{2} (A + B + C + D) + \frac{\gamma}{2} (- A k + B k - C k + D k), \\
 A + B - C - D & = - \frac{\overline{\gamma}}{2} (A + B + C + D) + \frac{\beta}{2} (- A k + B k - C k + D k), \\
 - A k e^{k d} + B k e^{- k d} & = 0, \\
 - C k e^{- k d} + D k e^{k d} & = 0.
\end{align*}
Furthermore, this set of equations allows a nontrivial choice of the coefficients $A, B, C, D$ if and only if
\begin{align*}
 \det \begin{pmatrix} 
	\frac{\alpha}{2} + k \big(1 - \frac{\gamma}{2} \big) & \frac{\alpha}{2} - k \big(1 - \frac{\gamma}{2} \big) & \frac{\alpha}{2} - k \big(1 + \frac{\gamma}{2} \big) & \frac{\alpha}{2} + k \big(1 + \frac{\gamma}{2} \big) \\
	- \frac{\beta}{2} k - \big(1 + \frac{\overline{\gamma}}{2} \big) & \frac{\beta}{2} k - \big(1 + \frac{\overline{\gamma}}{2} \big) & - \frac{\beta}{2} k + \big(1 - \frac{\overline{\gamma}}{2} \big) & \frac{\beta}{2} k + \big(1 - \frac{\overline{\gamma}}{2} \big) \\
	- k e^{k d} & k e^{- k d} & 0 & 0 \\
	0 & 0 & - k e^{- k d} & k e^{k d}
      \end{pmatrix} = 0,
\end{align*}
and the value of the determinant equals
\begin{align*}
 - \frac{1}{2} e^{- 2 k d} k^2 \Big( g (k) - h (k) j (k) \Big),
\end{align*}
which leads to~\eqref{eq:spectralIdentity}. 

Our aim is to estimate the largest solution $k > 0$ of~\eqref{eq:spectralIdentity}, which corresponds to the smallest eigenvalue of $H_{\cA, d}$. For this we distinguish several cases. Let us start with the case $\beta = 0$ (and $\alpha > 0$). In this case we have $j (k) = e^{2 k d} + e^{4 k d}$ and this function has no zeros. Thus~\eqref{eq:spectralIdentity} can be rewritten as
\begin{align}\label{eq:spectralIdentityNew}
 \frac{g (k)}{j (k)} = h (k).
\end{align}
Note that in this case
\begin{align*}
 \Big( \frac{g}{j} \Big)' (k) = - |\gamma|^2 e^{- 2 k d} (e^{2 k d} - 1 + 2 k d) < 0
\end{align*}
holds for all $k > 0$ and thus $g/j$ is strictly monotonously decreasing on $[0, \infty)$. Moreover, $(g/j) (0) = 0$ and $(g/j) (k) \to - \infty$ as $k \to + \infty$. On the other hand, the derivative of $h$ is given by
\begin{align*}
 h' (k) = 4 + 4 e^{- 2 k d} (2 k d + \alpha d - 1) > 4 - 4 e^{- 2 k d} > 0
\end{align*}
for each $k > 0$, hence $h$ is strictly monotonously increasing on $[0, \infty)$. Moreover, 
\begin{align}\label{eq:hProperties}
 h (0) = - 4 \alpha \quad \text{and} \quad \lim_{k \to + \infty} h (k) = + \infty. 
\end{align}
As $\alpha > 0$ it follows that~\eqref{eq:spectralIdentityNew} has precisely one positive solution. Note further that for each $k > 0$ and each $\eps \geq 0$ we have
\begin{align*}
 \Big(\frac{g}{j} - h \Big) (k  + \eps ) & = \frac{1}{e^{- 2 (k + \eps) d} + 1} \Big( e^{- 4 (k + \eps) d} \big( |\gamma|^2 (k + \eps) + 2 \alpha + 4 k + 4 \eps \big) \\
 & \qquad + 4 \alpha e^{- 2 (k + \eps) d} - |\gamma|^2 (k + \eps) + 2 \alpha - 4 k - 4 \eps \Big).
\end{align*}
On the one hand this implies
\begin{align*}
 \Big(\frac{g}{j} - h \Big) \Big( \frac{2 \alpha}{4 + |\gamma|^2} \Big) = 4 \alpha e^{\frac{- 4 \alpha}{4 + |\gamma|^2} d} > 0
\end{align*}
and, on the other hand,
\begin{align*}
 \lim_{d \to + \infty} \Big(\frac{g}{j} - h \Big) \Big( \frac{2 \alpha}{4 + |\gamma|^2} + \eps \Big) = - |\gamma|^2 \eps - 4 \eps < 0
\end{align*}
for each $\eps > 0$. Thus for each $\eps > 0$ there exists $d_0 > 0$ such that for each $d \geq d_0$ the only positive solution of~\eqref{eq:spectralIdentityNew} is contained in the interval $(\frac{2 \alpha}{4 + |\gamma|^2}, \frac{2 \alpha}{4 + |\gamma|^2} + \eps)$. This implies $m_{\cA, d} \leq m_\cA$ and $\lim_{d \to \infty} m_{\cA, d} = m_\cA$ in the case $\beta = 0$ and $\alpha > 0$.

Let us come to the case $\beta > 0$. We show first that in this case the equation~\eqref{eq:spectralIdentity} has either one or two positive solutions and that, if there are two, at each solution the function $g - h j$ has a sign change. For this we distinguish the two possibilities $\gamma = 0$ and $\gamma \neq 0$. If $\gamma = 0$ then the equation~\eqref{eq:spectralIdentity} is satisfied if and only if $k$ is a zero of either $j$ or $h$. Since $h$ is strictly monotonously increasing and satisfies~\eqref{eq:hProperties}, $h$ has precisely one positive zero if $\alpha > 0$ and no positive zero if $\alpha = 0$. Moreover, $j$ has precisely one positive zero $k_0$. Indeed, the equation $j (k) = 0$ can be rewritten as
\begin{align*}
 e^{- 2 k d} + 1 = - \frac{\beta}{2} k \big( e^{- 2 k d} - 1 \big),
\end{align*}
where the left-hand side is strictly monotonously decreasing on $[0, \infty)$ with values in $(1, 2]$, whilst the right-hand side is strictly monotonously increasing on $[0, \infty)$, taking the value 0 at $k = 0$ and the limit $+ \infty$ as $k \to + \infty$. Thus for $\gamma = 0$ the equation~\eqref{eq:spectralIdentity} has at most two solutions. Observe that $j$ is positive to the left of its zero and negative to the right. Hence, if the zeros of $j$ and $h$ do not coincide then $j h$ has a sign change at each of the two. In the case $\gamma \neq 0$ we can write~\eqref{eq:spectralIdentity} equivalently as~\eqref{eq:spectralIdentityNew} and the function $g$ is strictly monotonously decreasing with $g (k) < 0$ for all $k > 0$. Furthermore, for the unique positive zero $k_1$ of $j$ we have $(g/j) (k) < 0$ for $k \in (0, k_1)$ and $(g/j) (k) > 0$ for $k \in (k_1, \infty)$, hence
\begin{align*}
 \lim_{k \nearrow k_1} (g/j) (k) = - \infty \quad \text{and} \quad \lim_{k \searrow k_1} (g/j) (k) = + \infty.
\end{align*}
Moreover, 
\begin{align*}
 (g/j) (0) = 0 \quad \text{and} \quad \lim_{k \to + \infty} (g/j) (k) = \frac{2}{\beta} |\gamma|^2,
\end{align*}
as a simple calculation shows. Finally,
\begin{align*}
 (g/j)' (k) & = - 4 e^{-2 k d} |\gamma|^2 \frac{- 1 + e^{6 k d} + k d (2 + \beta k)}{\big(2 + \beta k + e^{2 k d} (2 - \beta k) \big)^2} \\
 & \quad - 4 e^{-2 k d} |\gamma|^2 \frac{e^{2 k d} (- 1 - 2 k d (\beta k - 2) ) + e^{4 k d} (1 + k d (\beta k + 2))}{\big(2 + \beta k + e^{2 k d} (2 - \beta k) \big)^2},
\end{align*}
which is negative for all $k \in (0, k_1) \cup (k_1, \infty)$. Thus $g/j$ is strictly monotonously decreasing on $(0, k_1) \cup (k_1, \infty)$. Comparing the properties of $g/j$ and $h$ it follows that there is precisely one positive solution of~\eqref{eq:spectralIdentity} if $\alpha = 0$ and precisely two positive solutions if $\alpha > 0$ and that $g/j - h$ and thus $g - h j$ changes its sign at each solution. Let now $\gamma$ be again arbitrary and let us estimate the largest solution. For each $k > 0$ and each $\eps \geq 0$ we have
\begin{align*}
 e^{- 4 (k + \eps) d} (g - j h) (k + \eps) = \sA e^{- 4 (k + \eps) d} + \sB e^{- 2 (k + \eps) d} + \sC
\end{align*}
with
\begin{align*}
 \sA = (\alpha + 2 (k + \eps)) (2 + \beta (k + \eps)) + (k + \eps) |\gamma|^2, \quad \sB = 4 (\alpha - \beta (k + \eps)^2)
\end{align*}
and
\begin{align*}
 \sC = -(\alpha - 2 (k + \eps)) (-2 + \beta (k + \eps)) - (k + \eps) |\gamma|^2.
\end{align*}
On the one hand for $\eps > 0$ this yields
\begin{align*}
 \lim_{d \to + \infty} e^{- 4 (k + \eps) d} (g - j h) (k + \eps) = \sC
\end{align*}
and thus for $k_0 = (4 + \det \cA + \sqrt{- 16 \alpha \beta + (4 + \det \cA)^2})/(4 \beta)$
\begin{align*}
 \lim_{d \to + \infty} e^{- 4 (k_0 + \eps) d} (g - j h) (k_0 + \eps) = \eps \big( 2 \beta \eps + \sqrt{- 16 \alpha \beta + (4 + \det \cA)^2} \big) > 0.
\end{align*}
On the other hand, we obtain
\begin{align*}
 e^{- 4 k_0 d} (g - j h) (k_0) & = \frac{(4 + \det \cA) \big(4 + \det \cA + \sqrt{- 16 \alpha \beta + (4 + \det \cA)^2} \big)}{2 \beta} e^{- 4 k_0 d} \\
 & \quad + \left( 4 \alpha - \frac{\big( 4 + \det \cA + \sqrt{- 16 \alpha \beta + (4 + \det \cA)^2} \big)^2}{4 \beta} \right) e^{- 2 k_0 d},
\end{align*}
which is negative for all sufficiently large $d$ since 
\begin{align*}
 16 \alpha \beta < ( 4 + \det \cA + \sqrt{- 16 \alpha \beta + (4 + \det \cA)^2} )^2
\end{align*}
whenever we are not in the case $\gamma = 0$ and $\alpha \beta = 4$. Thus for each $\eps > 0$ there exists $d_0 > 0$ such that for each $d \geq d_0$ the function $g - j h$ changes its sign in the interval $(k_0, k_0 + \eps)$. Observe that the so-obtained sequence of eigenvalues of $H_{\cA, d}$ does really represent the minima of the spectra. If not there would exist an accumulation point (since the minima are monotonous with respect to $d$ as one can see easily via the quadratic forms) $k_1 > k_0$ of solutions of~\eqref{eq:spectralIdentity} as $d \to \infty$. If $\eps_1 > 0$ and $\eps_2 > \eps_1$ are chosen such that $k_0 + \eps_1 < k_1 < k_0 + \eps_2$ then the above considerations yield that for each sufficiently large $d$ the signs of $g - j h$ at $k_0$ and $k_0 + \eps_1$ as well as at $k_0$ and $k_0 + \eps_2$ are opposite; in particular, for each sufficiently large $d$ the numbers $(g - j h) (k_0 + \eps_1)$ and $(g - j h) (k_0 + \eps_2)$ have the same sign. On the other hand, there is exactly one solution of~\eqref{eq:spectralIdentity} in $(k_1, k_0 + \eps_2) \subset (k_0 + \eps_1, k_0 + \eps_2)$ and the sign of $g - j h$ changes at that solution, which is a contradiction.

It remains to consider the case $\gamma = 0$ and $\alpha \beta = 4$. It is easy to see that in this case the zeros of $h$ and $j$ coincide for each $d > 0$. Thus the only positive solution of~\eqref{eq:spectralIdentity} in this case is given by the zero of, e.g., $h$. In this case $m_\cA = \alpha/2$ and we have $h (\alpha/2) = - 4 \alpha e^{- \alpha d} < 0$ as well as $\lim_{d \to 0} h (\alpha/2 + \eps) = 4 \eps > 0$ for each $\eps$ and the same reasoning as in the previous cases implies that the only negative eigenvalue of $H_{\cA, d}$ converges to $m_\cA$ as $d$ tends to $\infty$. This completes the proof of item~(ii) of the lemma.

If $d > 0$ is fixed then clearly the positive solutions of~\eqref{eq:spectralIdentity} depend continuously on $\alpha$ and $\beta$. Moreover, as $n_\tau$ converges to $1$ for $\tau \to \infty$, it follows from the above considerations that for each sufficiently large $\tau$ the number of positive solutions of~\eqref{eq:spectralIdentity} for $\widetilde \cA (\tau)$ and $\cA$ coincide (up to a possible crossing of solutions if $\gamma = 0$ and $\alpha \beta = 4$). Thus the claim of~(i) follows and the proof of Proposition~\ref{prop:schlimmesLemma} is complete.

\begin{acknowledgments}
This research was supported by the Czech Science Foundation (GA\v{C}R) within the project 14-06818S. J.R.\ gratefully acknowledges financial support by the Austrian Science Fund (FWF), project P 25162-N26, and the Austria--Czech Republic cooperation grant CZ01/2013.
\end{acknowledgments}

\end{document}